\documentclass[a4paper,UKenglish]{lipics}

\usepackage{microtype}%

\usepackage[utf8]{inputenc}
\usepackage{multicol}
\usepackage{graphicx}
\usepackage[usenames,dvipsnames]{xcolor}

\usepackage[fleqn]{amsmath}
\usepackage{amsfonts,amssymb}
\usepackage[natbib,cref]{nikis} 
		\RequirePackage{thmtools} %
		
		\newtheorem{Lemma}[theorem]{Lemma}
		\crefname{Theorem}{theorem}{theorems}
		\crefname{Lemma}{lemma}{lemmas}
		\Crefname{Theorem}{Theorem}{Theorems}
		\Crefname{Lemma}{Lemma}{Lemmas}

\usepackage{enumerate} %
\usepackage[shortlabels]{enumitem} %

\newcommand{\Wlogg}{Without loss of generality}

\newcommand{\intervalI}{\mathcal{I}} %
\renewcommand{\b}[1]{\mathsf{b}(#1)}%
\newcommand{\e}[1]{\mathsf{e}(#1)}%

\renewcommand{\r}[1]{\ensuremath{{#1}_{\mathup{\rho}}}}
\newcommand{\arm}   {\ensuremath{u}}
\newcommand{\subarm}{\ensuremath{s}}
\newcommand{\rep}   {\ensuremath{r}}
\newcommand{\gap}   {\ensuremath{v}}
\newcommand{\gapS}   {\ensuremath{g}}
\newcommand{\zarm}  {\ensuremath{z}}

\newcommand{\agr}  {\ensuremath{\sigma}} %

\newcommand{\pali}[1]{\ensuremath{{#1}^{\mathup{\intercal}}}}

\newcommand{\grGR}  [1][w]{\ensuremath{\mathcal{G }_\alpha(#1)}}

\newcommand{\grP}   [1][w]{\ensuremath{\beta\mathcal{P }_\alpha(#1)}}
\newcommand{\grNP}  [1][w]{\ensuremath{\overline{\beta\mathcal{P }}_\alpha(#1)}}
\newcommand{\gpP}   [1][w]{\ensuremath{\pali{\beta\mathcal{P }}_\alpha(#1)}}
\newcommand{\gpNP}  [1][w]{\ensuremath{\pali{\overline{\beta\mathcal{P }}}_\alpha(#1)}}

\newcommand{\gpGR}  [1][w]{\ensuremath{\pali{\mathcal{G }}_\alpha(#1)}}

\newcommand{\period}{p}

\newcommand{\s}[1]{\overline{#1}}

\usepackage{tikz}
\usetikzlibrary{patterns}
\tikzstyle{arm} = [pattern=horizontal lines, pattern color=blue!10]
\tikzstyle{subarm} = [pattern=vertical lines, pattern color=green!10]
\tikzstyle{gap} = [pattern=north east lines, pattern color=red!10]
\tikzstyle{rep} = [pattern=north west lines, pattern color=purple!10]
\tikzstyle{diff} = [pattern=fivepointed stars, pattern color=cyan!10]

\edef\ourCol{0}
\newcount\ourRow%
\pgfmathsetcount{\ourRow}{0}

\newcommand{\col}[3][]{%
\draw [#1] (\ourCol,\ourRow) rectangle node {#3} (\ourCol+#2,\ourRow-1);
\edef\ourCol{\ourCol+#2}
}
\newcommand{\void}[1]{%
\edef\ourCol{\ourCol+#1}
}
\newcommand{\newrow}{%
\pgfmathsetcount{\ourRow}{\ourRow-1}
\edef\ourCol{0}
}
\newcommand{\lab}[1]{%
\draw (\ourCol,\ourRow) -- (\ourCol,\ourRow-1) node [anchor=north] {#1};
}

\usepackage{amssymb}
\usepackage{amsmath}
\usepackage{latexsym}
\usepackage{verbatim}
\usepackage{amsfonts}%
\usepackage{graphicx}
\usepackage{latexsym,graphicx}
\usepackage{xcolor}
\usepackage{url}
\usepackage{algorithm}
\usepackage{algorithmic}

\newcommand{\per}{{per}}

\newcommand{\bigo}{{\mathcal O}}

\newcommand{\LCP}{{\mathit{LCP}}}

\title{Efficiently Finding All Maximal $\alpha$-gapped Repeats}
\titlerunning{Efficiently Finding All Maximal $\alpha$-gapped Repeats} %

\author[1]{Pawe\l{} Gawrychowski}
\author[2]{Tomohiro I}
\author[3]{Shunsuke Inenaga} 
\author[2]{Dominik K\"{o}ppl}
\author[4]{Florin Manea}
\affil[1]{
Institute of Informatics, University of Warsaw, Poland,\\ 
gawry@mimuw.edu.pl
}
\affil[2]{
Department of Computer Science, TU Dortmund, Germany\\
\{tomohiro.i,dominik.koeppl\}@cs.tu-dortmund.de}
\affil[3]{Department of Informatics, Kyushu University, Japan\\
inenaga@inf.kyushu-u.ac.jp}
\affil[4]{
Department of Computer Science, Kiel University, Germany\\
flm@informatik.uni-kiel.de}

\authorrunning{Gawrychowski et al.} %

\Copyright{}%

\serieslogo{}%
\volumeinfo%
  {}%
  {2}%
  {}%
  {1}%
  {1}%
  {1}%
\EventShortName{}
\DOI{}%

\begin{document}

\maketitle
\renewcommand{\l}[1]{\ensuremath{{#1}_{\mathup{\lambda}}}}

\begin{abstract}
For $\alpha\geq 1$, an $\alpha$-gapped repeat in a word $w$ is a factor $uvu$ of $w$ such that $|uv|\leq \alpha |u|$; the two factors $u$ in such a repeat are called arms, while the factor $v$ is called gap. Such a repeat is called maximal if its arms cannot be extended simultaneously with the same symbol to the right or, respectively, to the left. In this paper we show that the number of maximal $\alpha$-gapped repeats that may occur in a word is upper bounded by $18\alpha n$. This allows us to construct an algorithm finding all the maximal $\alpha$-gapped repeats of a word in $\bigo(\alpha n)$; this is optimal, in the worst case, as there are words that have $\Theta(\alpha n)$ maximal $\alpha$-gapped repeats. Our techniques can be extended to get comparable results in the case of $\alpha$-gapped palindromes, i.e., factors $uv\pali{u}$ with $|uv|\leq \alpha |u|$.
 \end{abstract}

\section{Introduction}
Gapped repeats and palindromes are repetitive structures occurring in words that were investigated extensively within theoretical computer science (see, e.g., \cite{Gu97,Brodal,KK_SPIRE,gappedPalindroms,KolpakovPPK14,power_of_SA,Cro2011,KKC15,FCT,MFCS,TanimuraFIIBT15} and the references therein) with motivation coming especially from the analysis of DNA and RNA structures, where they were used to model different types of tandem and interspersed repeats as well as hairpin structures; such structures are important in analyzing the structural and functional information of the genetic sequences (see, e.g.,  \cite{Gu97,Brodal,gappedPalindroms}).

Following~\cite{gappedPalindroms,KolpakovPPK14}, we study gapped repeats (palindromes) $uvu$ (respectively, $uv\pali{u}$) where the length of the gap $v$ is upper bounded by the length of the arm $u$ multiplied by some factor, also known as \intWort{$\alpha$-gapped repeats} and \intWort{$\alpha$-gapped palindromes} occurring in a word. 

The work on $\alpha$-gapped palindromes was focused so far on combinatorial and algorithmic problems that extend the classical results obtained for squares and palindromes. Namely, problems like how many maximal $\alpha$-gapped repeats or palindromes does a word of length $n$ contain (here maximal means that the arms of the repeat cannot be both extended to the right or left with the same symbol), how efficiently can we compute the set of maximal $\alpha$-gapped repeats or palindromes of a word, how efficiently can we compute the longest $\alpha$-gapped repeat or palindrome, were investigated (see, for instance, \cite{gappedPalindroms,Brodal,KolpakovPPK14,FCT,TanimuraFIIBT15,KKC15}, and the references therein). In this paper we obtain the following results:
\begin{itemize}
	\item The number of maximal $\alpha$-gapped repeats in a word of length $n$ is at most $18\alpha n$. 
	\item We can compute the list of all $\alpha$-gapped repeats in $\Oh{\alpha n}$ time for \emph{integer} alphabets.
\end{itemize} 
Our techniques can be extended to show that the number of maximal $\alpha$-gapped palindromes in a word of length $n$ is upper bounded by $28\alpha n + 7n$; they can be found in $\bigo(\alpha n)$ time. As there are words of length $n$ that contain $\Theta(\alpha n)$ maximal $\alpha$-gapped repeats or palindromes (see \cite{KolpakovPPK14}), it follows that the bounds on the number of maximal $\alpha$-gapped repeats or palindromes we obtained are asymptotically tight, and that we cannot hope for algorithms finding all $\alpha$-gapped repeats or palindromes faster in the worst case.

Our results improve those of~\cite{KolpakovPPK14} (as well as those existing in the literature before~\cite{KolpakovPPK14}), where the size of the set of maximal $\alpha$-gapped repeats with non-empty gap is shown to be $\bigo(\alpha^2 n)$, and can be computed in $\bigo(\alpha^2 n)$ time for integer alphabets. An alternative proof of the fact that the number of maximal $\alpha$-gapped repeats in a word of length $n$ is $\bigo(\alpha n)$ was given in the very recent \cite{KKC15}; however, compared to the respective paper, we give a more direct proof of the $\bigo(\alpha n)$ upper bound as well as a concrete evaluation of the constant hidden by the $\bigo$-denotation. In~\cite{KKC15,TanimuraFIIBT15} algorithms producing all the maximal $\alpha$-gapped repeats of a word were given; in the light of the upper bound $\bigo(\alpha n)$ on the number of maximal $\alpha$-gapped repeats that may occur in a word of length $n$, it follows that these algorithms work in $\bigo(\alpha n)$ time, but only for constant alphabets. Extending the approach in \cite{FCT}, we show here that, in fact, such algorithms can be also designed for integer alphabets; our algorithm requires a deeper analysis than the one developed in \cite{FCT} for finding the longest $\alpha$-gapped repeat, and uses essentially different techniques and data structures than the ones in~\cite{KKC15,TanimuraFIIBT15}.

 \section{Combinatorics on words preliminaries}
Let $\Sigma$ be a finite alphabet; $\Sigma^*$ denotes the set of all finite words over $\Sigma$. The \intWort{length} of a word $w\in \Sigma^*$ is denoted by $\left|w\right|$. The \intWort{empty word} is denoted by ${\varepsilon}$. 
A word $u\in \Sigma^*$ is a \intWort{factor} of $v\in \Sigma^*$ if $v=xuy$, for some $x, y\in \Sigma^*$; we say that $u$ is a \intWort{prefix} of $v$, if $x={\varepsilon}$, and a \intWort{suffix} of $v$, if $y={\varepsilon}$.
We denote by $w[i]$ the symbol occurring at position $i$ in $w,$ and by $w[i,j]$ the factor of $w$ starting at position $i$ and ending at position $j,$ consisting of the catenation of the symbols $w[i], \ldots, w[j],$ where $1\leq i\leq j\leq n$; we define $w[i,j]=\varepsilon$ if $i>j$. {By ${\pali{w}}$ we denote the mirror image of $w$.} A \intWort{period} of a word $w$ over $\Sigma$ is a positive integer $p$ such that $w[i]=w[j]$ for all $i$ and $j$ with $i\equiv j\pmod{p}$; a word that has period $p$ is also called \intWort{$p$-periodic}. Let $\per(w)$ be the smallest~period~of~$w$. A word $w$ with $\per(w)\leq\frac{|w|}{2}$ is called \intWort{periodic}; otherwise, $w$ is called \intWort{aperiodic}. It is worth noting that the length of the overlap between two consecutive occurrences of an aperiodic factor $v$  in $w$ is upper bounded by $\frac{|v|}{2}$.

By $\intervalI=[b,e]$ we represent the set of consecutive integers from $b$ to $e$, for $b\le e$, and call $\intervalI$ an \intWort{interval}. 
For an interval $\intervalI$, we use the notations $\b{\intervalI}$ and $\e{\intervalI}$ to denote the beginning and end of $\intervalI$;
i.e., $\intervalI = [\b{\intervalI},\e{\intervalI}]$.
We write $\abs{\intervalI}$ to denote the length of $\intervalI$; i.e., $\abs{\intervalI}=\e{\intervalI}-\b{\intervalI}+1$.
A \intWort{subword}~$\arm$ of a word $w$ is a pair $(\subarm, [b, e])$ consisting of a factor $\subarm$ of $w$ and an interval $[b, e]$ in $w$ 
such that $\subarm = w[b,e]$.
While a factor is identified only by a sequence of letters, a subword is also identified by its position in the word.
So subwords are always unique, while a word may contain multiple occurrences of the same factor.
For two subwords $\arm$ and $\s\arm$ of a word~$w$, we write $\arm = \s\arm$ 
if they start at the same position in~$w$ and have the same length.
We write $\arm \equiv \s\arm$ if the factors identifying these subwords are the same. %
We implicitly use subwords both like factors of $w$ 
and as intervals contained in $[1,\abs{w}]$, e.g., 
we write $\arm \subseteq \s\arm$ if
two subwords $\arm = (\subarm, [b, e]) , \s\arm = (\s\subarm, [\s{b}, \s{e}])$ of $w$ satisfy $[b, e] \subseteq [\s{b}, \s{e}]$, i.e.,
$\b{\s\arm} \le \b{\arm} \le \e{\arm} \le \e{\s\arm}$.
Two subwords $\arm$ and $\s\arm$ of the same word~$w$ are called \intWort{consecutive}, iff $\e{\arm}+1 = \b{\s\arm}$.%
For a word $w$, we call a triple of consecutive subwords $\l\arm,\gap,\r\arm$ a \intWort{gapped repeat} 
with period $\abs{\l\arm\gap}$ and gap $|v|$ iff $\r\arm \equiv \l\arm$.
A triple of consecutive subwords $\l\arm,\gap,\r\arm$ is called a \intWort{gapped palindrome} with gap $\abs{\gap}$ iff $\r\arm \equiv \pali{\l\arm}$.
The subwords $\l\arm$ and $\r\arm$ are called left and right \intWort{arm}, respectively.
For $\alpha\geq 1$, the gapped repeat (palindrome) $\l\arm,\gap,\r\arm$ is called \intWort{$\alpha$-gapped} iff $\abs{\l\arm}+\abs{\gap} \le \alpha \abs{\l\arm}$.
Further, it is called \intWort{maximal} iff %
that
$w[\b{\l\arm}-1] \not= w[\b{\r\arm}-1]$ and 
$w[\e{\l\arm}+1] \not= w[\e{\r\arm}+1]$,
and for a gapped palindrome $\l\arm,\gap,\r\arm$ that
$w[\b{\l\arm}-1] \not= w[\e{\r\arm}+1]$ and 
$w[\e{\l\arm}+1] \not= w[\b{\r\arm}-1]$.
Let $\grGR$ (respectively, $\gpGR$) denote the set of maximal $\alpha$-gapped repeats (palindromes) in~$w$.
The representation of a maximal gapped repeat (palindrome) by the subword $\zarm := w[\l\arm]w[\gap]w[\r\arm]$ is not unique ---
the same subword~$\zarm$ can be composed of gapped repeats (palindromes) with different periods (different gaps).
Instead, a maximal gapped repeat (palindrome) is uniquely determined by its left arm~$\l\arm$ and its period (gap).
By fixing~$w$, we therefore can map $\l\arm,\gap,\r\arm$ injectively to the pair of integers $(\e{\l\arm}, \abs{\l\arm\gap})$ in case of a gapped repeat, 
or to $(\e{\l\arm}, \abs{\gap})$ in case of a gapped palindrome.

A run in a word $w$ is a maximal periodic factor; the exponent of a run is the number of times the period fits in that run. For a word $w$, 
let $R(w)$ and $E(w)$ denote the number of runs and the sum of the exponents of runs in~$w$, respectively.
The exponent of a run~$\rep$ is denoted by $\exp(\rep)$.
We use the following results from literature:

\begin{Lemma}[\cite{runstheorem}]\label{lemmaExponent}
	For a word $w$, $E(w) < 3\abs{w}$.
\end{Lemma}

\begin{Lemma}[{\cite{KolpakovPPK14}}]\label{lemmaRepetitionsOverlap}
Two distinct maximal repetitions with the same minimal period~$p$ cannot
have an overlap of length greater than or equal to~$p$.
\end{Lemma}

\begin{corollary}[{\cite{KolpakovPPK14}}]\label{lemmaNumCycleSquares}
If a square $uu$ is primitive, any word~$v$ contains no more than $|v|/|u|$
occurrences of $uu$.
\end{corollary}

\begin{Lemma}
	Inverting a gapped repeat (palindrome) results in a gapped repeat (palindrome) with the same period. 
	Hence there exist the bijections
	$\grGR \sim \grGR[\pali{w}]$ and $\gpGR \sim \gpGR[\pali{w}]$.
\end{Lemma}

\section{Point analysis}\label{secPointAnalysis}

A pair of positive integers is called a \intWort{point}.
We use points to bound the cardinality of a subset of gapped repeats and gapped palindromes
by injectively mapping a gapped repeat (palindrome) to a point as stated above. 
To this end, we show that some vicinity of any point generated by a member of this subset does not contain any
point that is generated by another member. This vicinity is given by
\begin{definition}
  For any $\gamma \in (0,1]$, we say that a point $(x, y)$ $\gamma$-\intWort{covers} a point $(x', y')$
  iff $x - \gamma y \le x' \le x$ and $y - \gamma y \le y' \le y$.
\end{definition}
It is crucial that the $\gamma$~factor is always multiplied with the $y$-coordinates.
In other words, the number of $\gamma$-covers of a point $(\cdot,y)$ correlates with $\gamma$ and the value~$y$.
The main property of this definition is given by

\begin{Lemma}\label{lemmaPoints}
For any $\gamma \in (0,1]$,
let $S \subset {[1,n]}^2$ be a set of points such that no two distinct points in~$S$ $\gamma$-cover the same point.
  Then $\abs{S} < 3 n / \gamma$.
\end{Lemma}
\begin{proof}
	We estimate the maximal number of points that can be placed in ${[1,n]}^2$ such that their covered points are disjoint.
	First, the number of points $(\cdot,y) \in {[1,n]}^2$ with $y < 1/\gamma$ is less than $n / \gamma$.
	Second, if a point $(\cdot, y)$ satisfies $2^l / \gamma \leq y < 2^{l+1} / \gamma$ for some integer $l \geq 0$,
	the point~$(\cdot, y)$ $\gamma$-covers at least $2^l \times 2^l$ points, or to put it differently,
	this point $\gamma$-covers at least $2^l$ points $(\cdot,y')$ with $y - 2^l \le y' \le y$.
	In other words, there are at most $n / (2^l \gamma)$ points in $S$ with $2^l / \gamma \leq y < 2^{l+1} / \gamma$.
  Hence, $\abs{S} < n / \gamma + \sum_{l=0}^\infty n / (2^l \gamma) = 3 n / \gamma$.
\end{proof}

Kolpakov et al.~\cite{KolpakovPPK14} split the set of maximal $\alpha$-gapped repeats into three subsets, and studied the maximal size of each subset.
They analyzed maximal $\alpha$-gapped repeats by partitioning them into three 
subsets: 
\sitemize*{%
	\item those contained in some repetition,
	\item those having arms containing a periodic prefix or suffix that is larger than half of the size of the arms, 
	\item and those not belonging to both former subsets.
} 

They showed that the first two subsets contain at most $\Oh{\alpha n}$ elements.
The point analysis is used as a tool for studying the last subset.
By mapping a gapped repeat to a point consisting of the end position of its left arm and its period,
they showed that the points created by two different maximal $\alpha$-gapped repeats cannot $\frac{1}{4\alpha}$-cover the same point.
By this property, they bounded the size of the last subset by $\Oh{\alpha^2 n}$.
\Cref{lemmaPoints} immediately improves this bound of $\Oh{\alpha^2 n}$ to $\Oh{\alpha n}$. Consequently, it shows that the number of maximal $\alpha$-gapped repeats of a word of length $n$ is $\Oh{\alpha n}$.

\section{Gapped repeats}

We optimize the proof technique from~\cite{KolpakovPPK14} and 
improve the upper bound of the number of maximal $\alpha$-gapped repeats in a word of length $n$ from $\Oh{\alpha n}$ to $18 \alpha n$.
Unlike~\cite{KolpakovPPK14,KKC15}, we partition the maximal $\alpha$-gapped repeats differently.
We categorize a gapped repeat (palindrome) depending on whether their left arm contains a periodic prefix or not.
Both subsets are treated differently.
For the ones having an periodic prefix, we think about the number of runs covering this prefix.
The other category is analyzed by using the results of \Cref{secPointAnalysis}.
We begin with a formal definition of both subsets and analyze the former subset.

Let $0 < \beta < 1$.
A gapped repeat (palindrome) $\agr = \l\arm,\gap,\r\arm$ belongs to \grP{} (\gpP{}) iff \l{\arm} contains 
a periodic prefix of length at least $\beta\abs{\l{\arm}}$.
We call $\agr$ \intWort{periodic}.
Otherwise $\agr \in \grNP$ ($\agr \in \gpNP$), where $\grNP := \grGR \setminus \grP$ and $\gpNP := \gpGR \setminus \gpP$;
we call $\agr$ \intWort{aperiodic}.

\begin{Lemma}\label{lemmaP}
Let $w$ be a word, $\alpha > 1$ and $0 < \beta < 1$ 
two real numbers.
Then $\abs{\grP}$ is at most $2 \alpha E(w) / \beta$. 
\end{Lemma}
\begin{proof}
  Let $\agr = (\l\arm,\gap,\r\arm) \in \grP$.
  By definition, the left arm~$\l\arm$ has a periodic prefix~$\l\subarm$ of length at least $\beta \abs{\l\arm}$.
  Let $\l\rep$ denote the run that generates $\l\subarm$, i.e.,
  $\l\subarm \subseteq \l\rep$ and they both have the common shortest period $\period$.
  By the definition of gapped repeats, there is a right copy $\r\subarm$ of $\l\subarm$ contained in $\r\arm$ with
  $%
	  \r\subarm = 
		  w[\b{\l\subarm}+\abs{\l\arm\gap},\e{\l\subarm}+\abs{\l\arm\gap}] \equiv \l\subarm. 
  $%
  
  Let $\r\rep$ be a run generating $\r\subarm$ (it is possible that $\r\rep$ and $\l\rep$ are identical).
  By definition, $\r\rep$ has the same period~$\period$ as $\l\rep$.
  In the following, we will see that $\agr$ is uniquely determined by $\l\rep$ and 
	  the period $q := \abs{\l\arm\gap}$, if $\agr$ is a periodic gapped repeat.
  We will fix $\l\rep$ and pose the question how many maximal periodic gapped repeats can be generated by $\l\rep$.
  
  Since $\agr$~is maximal, $\b{\l\arm} = \b{\l\rep}$ or $\b{\r\arm} = \b{\r\rep}$ must hold;
  otherwise we could extend $\agr$ to the left. 
  We analyze the case $\b{\l\subarm} = \b{\l\rep}$, the other is treated exactly in the same way by symmetry.
    The gapped repeat~$\agr$ is identified by~$\l\rep$ and the period~$q$.
    We fix $\l\rep$ and count the number of possible values for the period~$q$.
    Given two different gapped repeats $\agr_1$ and $\agr_2$ with respective periods $q_1$ and $q_2$ such that the left arms of both
	are generated by $\l\rep$, the difference~$\delta$ between $q_1$ and $q_2$ must be at least $\period$.

    Since $\abs{\l\arm} \leq \abs{\l\subarm} / \beta$ and $\agr$ is $\alpha$-gapped,
    $1 \leq q \leq \abs{\l\subarm} \alpha / \beta \leq \abs{\l\rep} \alpha / \beta$.
    Then the number of possible periods $q$ is bounded by 
    $\abs{\l\rep} \alpha / (\beta \period) = \exp(\l\rep) \alpha / \beta$.
    Therefore the number of maximal $\alpha$-gapped repeats is bounded by $\alpha E(w) / \beta$ for the case $\b{\l\arm} = \b{\l\rep}$.
  Summing up we get the bound $2 \alpha E(w) / \beta$.
  \end{proof}

Remembering the results of \Cref{secPointAnalysis}, we map gapped repeats  to their respective points.
By using the period as the $y$-coordinate, one can show \Cref{lemmaNPCover}.

\begin{Lemma}\label{lemmaNPCover}
	Given a word~$w$, and two real numbers $\alpha > 1$ and $2/3 \leq \beta < 1$.
	The points mapped by two different maximal gapped repeats in~$\grNP$ cannot $\frac{1-\beta}{\alpha}$-cover the same point.
\end{Lemma}
\begin{proof}
	Let $\agr = \l\arm,\gap,\r\arm$ and $\s\agr = \s{\l\arm},\s{\gap},\s{\r\arm}$ be two different maximal gapped repeats in~$\grNP$.
  Set $\arm := \abs{\l\arm} = \abs{\r\arm}$, $\s\arm := \abs{\s{\l\arm}} = \abs{\s{\r\arm}}$, $q := \abs{\l\arm\gap}$ and $\s{q} := \abs{\s{\l\arm}\s{\gap}}$.
  We map the maximal gapped repeats~$\agr$ and~$\s\agr$ to the points~$(\e{\l\arm}, q)$ and~$(\e{\s{\l\arm}},\s{q})$, respectively.
  Assume, for the sake of contradiction, that both points $\frac{1-\beta}{\alpha}$-cover the same point $(x, y)$.

  Let $\zarm := \abs{\e{\l\arm} - \e{\s{\l\arm}}}$ be the difference of the endings of both left arms, and
  $\l\subarm := w[ [\b{\l\arm},\e{\l\arm}]\cap [\b{\s{\l\arm}},\e{\s{\l\arm}}] ]$ 
  be the overlap of $\l\arm$ and $\s{\l\arm}$.
  Let $\subarm := \abs{\l\subarm}$, and
  let $\r\subarm$ (resp. $\s{\r\subarm}$) be the right copy of $\l\subarm$ based on $\agr$ (resp. $\s\agr$).

  {\bf Sub-Claim:}  The overlap~$\l\subarm$ is not empty, and $\r\subarm \not= \s{\r\subarm}$

  {\bf Sub-Proof.}
  Assume for this sub-proof that $\e{\l\arm} < \e{\s{\l\arm}}$ (otherwise exchange $\agr$ with~$\s\agr$, or yield the contradiction $\agr = \s\agr$). 
  By combining the $(1-\beta)/\alpha$-cover property with the fact that $\s\agr$ is $\alpha$-gapped, we yield
  \(
	  \e{\s{\l\arm}} - \s\arm \le \e{\s{\l\arm}} - \s{q}(1-\beta)/\alpha \le x \le \e{\l\arm} < \e{\s{\l\arm}}.
  \)
  So the subword $w[\e{\l\arm}]$ is contained in $\s{\l\arm}$.
  If $\r\subarm = \s{\r\subarm}$, then we get a contradiction to the maximality of $\agr$:
  By the above inequality, $w[\e{\l\arm}+1]$ is contained in $\s{\l\arm}$, too. 
  Since $\s\agr$ is a gapped repeat, the character~$w[\e{\l\arm}+1]$ occurs in $\s{\r\arm}$, exactly at $w[\e{\r\arm}+1]$.
  \qed{}

  \Wlogg{} let $q \le \s{q}$.  Then
  \begin{flalign}
	  \label{equOrdGapRuleInv}
	  &\s{q} - \frac{\s{q}(1-\beta)}{\alpha} \le y \le q \le \s{q}. 
  &\\
  \label{equOrdGapRule}
	  &\text{So the difference of both periods is }
    0 \le \delta := \s{q} - q \le \s{q}(1-\beta)/\alpha \le \s\arm(1-\beta).
	&\\
  \label{equOrdGapRuleTakeFirstArm}
  &\text{\Cref{equOrdGapRuleInv} also yields that }
  \arm \ge q/\alpha \ge \frac{\s{q}}{\alpha} (1 - \frac{1-\beta}{\alpha}) \ge \s{q}\beta/\alpha.
  &
  \end{flalign}
  Since $\r\subarm = [\b{\l\subarm}+q,\e{\l\subarm}+q]$ 
  and $\s{\r\subarm} = [\b{\l\subarm}+\s{q},\e{\l\subarm}+\s{q}]$, we have
  $\b{\s{\r\subarm}} - \b{\r\subarm} = \delta$.
  
  By case analysis, we show that $\l\arm$ or $\s{\l\arm}$ has a periodic prefix,
  which leads to the contradiction that $\agr$ or $\s\agr$ are in $\grP$.

  {\bf 1. Case $\e{\l\arm} \le \e{\s{\l\arm}}$}.
  Since $\e{\s{\l\arm}} - \s{q}(1-\beta)/\alpha \le x \le \e{\l\arm} \le \e{\s{\l\arm}}$,
  \begin{equation}\label{equOrdGapRulePosFirstLeft}
    \zarm = \e{\s{\l\arm}} - \e{\l\arm} \le \s{q}(1-\beta)/\alpha \le \s\arm(1-\beta).
  \end{equation}

  \begin{figure}[h]
    \centering{%
		\begin{tikzpicture}[yscale=0.4,xscale=0.92]
  \col[arm]{3}{$\l\arm$}
  \col[gap]{2}{$\gap$}
  \col[arm]{3}{$\r\arm$}
  \newrow
  \void{0.5}
  \col[arm]{3.5}{$\s{\l\arm}$}
  \col[gap]{2.5}{$\s\gap$}
  \col[arm]{3.5}{$\s{\r\arm}$}
  \newrow
  \void{0.5}
  \col[subarm]{2.5}{$\l\subarm$}
  \col[diff]{1}{$\zarm$}
  \void{1.5}
  \col[subarm]{2.5}{$\r\subarm$}
  \newrow
  \void{5.5}
  \col[rep]{1}{$\delta$}
  \col[subarm]{2.5}{$\s{\r\subarm}$}
	\end{tikzpicture}
    }
	\caption{Sub-Case 1a}
  \end{figure}
  {\bf 1a. Sub-Case $\b{\l\arm} \le \b{\s{\l\arm}}$}.
  By~\Cref{equOrdGapRulePosFirstLeft}, we get $\subarm = \s\arm - \zarm \ge \s\arm \beta$.
  It follows from~\Cref{equOrdGapRule} and $2/3 \le \beta < 1$ that
  $\subarm / \delta \ge \s\arm \beta / \s\arm (1-\beta) = \beta / (1-\beta) \ge 2$,
  which means that $\r\subarm$ and $\s{\r\subarm}$ overlap at least half of their common length, so $\l\subarm$ is periodic.
  Since $\l\subarm$ is a prefix of $\s{\l\arm}$ of length $\subarm \ge \s\arm \beta$, $\s\agr$ is in $\grP$, a contradiction.

  \begin{figure}[h]
    \centering{%
		\begin{tikzpicture}[yscale=0.4,xscale=0.92]
			\void{0.5}
  \col[arm]{3}{$\l\arm$}
  \col[gap]{2}{$\gap$}
  \col[arm]{3}{$\r\arm$}
  \newrow
  \col[arm]{4}{$\s{\l\arm}$}
  \col[gap]{2.5}{$\s\gap$}
  \col[arm]{4}{$\s{\r\arm}$}
  \newrow
  \void{0.5}
  \col[subarm]{3}{$\l\subarm$}
  \col[diff]{0.5}{$\zarm$}
  \void{1.5}
  \col[subarm]{3}{$\r\subarm$}
  \newrow
  \void{5.5}
  \col[rep]{1.5}{$\delta$}
  \col[subarm]{3}{$\s{\r\subarm}$}
	\end{tikzpicture}
    }
	\caption{Sub-Case 1b}
  \end{figure}
  {\bf 1b. Sub-Case $\b{\l\arm} > \b{\s{\l\arm}}$}.
  We conclude that $\l\subarm = \l\arm$.
  It follows from~\Cref{equOrdGapRule,equOrdGapRuleTakeFirstArm} and $2/3 \le \beta < 1$ that
  $\subarm / \delta \ge \s{q} \alpha \beta / (\s{q} \alpha (1-\beta)) = \beta / (1-\beta) \ge 2$,
  which means that $\l\subarm = \l\arm$ is periodic.
  Hence $\agr$ is in $\grP$, a contradiction.

  {\bf 2. Case $\e{\l\arm} > \e{\s{\l\arm}}$}.
  Since $\e{\l\arm} - q(1-\beta)/\alpha \le x \le \e{\s{\l\arm}} \le \e{\l\arm}$,
  \begin{equation}\label{equOrdGapRulePosFirstRight}
    \zarm = \e{\l\arm} - \e{\s{\l\arm}} \le q(1-\beta)/\alpha \le \s{q}(1-\beta)/\alpha \le \s\arm(1-\beta).
  \end{equation}

  \begin{figure}[h]
    \centering{%
		\begin{tikzpicture}[yscale=0.4,xscale=0.92]
  \col[arm]{3}{${\l\arm}$}
  \col[gap]{2.5}{$\gap$}
  \col[arm]{3}{${\r\arm}$}
  \newrow
  \void{0.5}
  \col[arm]{2}{$\s{\l\arm}$}
  \col[gap]{4.5}{$\s\gap$}
  \col[arm]{2}{$\s{\r\arm}$}
  \newrow
  \void{0.5}
  \col[subarm]{2}{$\l\subarm$}
  \col[diff]{0.5}{$\zarm$}
  \void{3}
  \col[subarm]{2}{$\r\subarm$}
  \newrow
  \void{6}
  \col[rep]{1}{$\delta$}
  \col[subarm]{2}{$\s{\r\subarm}$}
	\end{tikzpicture}
    }
	\caption{Sub-Case 2a}
  \end{figure}
  {\bf 2a. Sub-Case $\b{\l\arm} \le \b{\s{\l\arm}}$}.
  We conclude that $\l\subarm = \s{\l\arm}$.
  It follows from~\Cref{equOrdGapRule} and $2/3 \leq \beta < 1$ that
  $\subarm / \delta \ge \s\arm / (\s\arm (1-\beta)) = 1 / (1-\beta) \ge 3 > 2$,
  which means that $\l\subarm = \s{\l\arm}$ is periodic.
  Hence $\s\agr$ is in $\grP$, a contradiction.

  \begin{figure}[h]
    \centering{%
		\begin{tikzpicture}[yscale=0.4,xscale=0.92]
  \void{0.5}
  \col[arm]{3.5}{${\l\arm}$}
  \col[gap]{2.5}{$\gap$}
  \col[arm]{3.5}{${\r\arm}$}
  \newrow
  \col[arm]{3}{$\s{\l\arm}$}
  \col[gap]{4.5}{$\s\gap$}
  \col[arm]{3}{$\s{\r\arm}$}
  \newrow
  \void{0.5}
  \col[subarm]{2.5}{$\l\subarm$}
  \col[diff]{1}{$\zarm$}
  \void{2.5}
  \col[subarm]{2.5}{$\r\subarm$}
  \newrow
  \void{6.5}
  \col[rep]{1.5}{$\delta$}
  \col[subarm]{2.5}{$\s{\r\subarm}$}
	\end{tikzpicture}
    }
	\caption{Sub-Case 2b}
  \end{figure}
  {\bf 2b. Sub-Case $\b{\l\arm} > \b{\s{\l\arm}}$}.
  By~\Cref{equOrdGapRulePosFirstRight}, we get $\subarm = \arm - \zarm \ge \arm \beta$.
  If $\delta \le \subarm / 2$, $\r\subarm$ and $\s{\r\subarm}$ overlap at least half of their common length,
  which leads to the contradiction that $\l\arm$ has a periodic prefix $\l\subarm$ of length at least $\arm \beta$.
  Otherwise, let us assume that $\subarm / 2 < \delta$.
  By~\Cref{equOrdGapRule,equOrdGapRuleTakeFirstArm} we get
  $\arm / \delta \ge \s{q} \alpha \beta / (\s{q} \alpha (1-\beta)) = \beta / (1-\beta) \ge 2$
  with $2/3 \le \beta < 1$.
  Hence, $\delta$ is upper bounded by $\arm / 2$; so
  $\r\arm$ has a periodic prefix of length at least $2 \delta$ (since $2 \delta > \subarm \ge \arm \beta$), a contradiction.
\end{proof}

The next \namecref{lemmaNP} follows immediately from~\Cref{lemmaPoints,lemmaNPCover}.
\begin{Lemma}\label{lemmaNP}
  For $\alpha > 1$, $2/3 \leq \beta < 1$ and a word $w$ of length $n$, $\abs{\grNP} < 3 \alpha n / (1 - \beta)$.
\end{Lemma}

\begin{theorem}
	Given a word~$w$ of length~$n$, and a real number $\alpha > 1$.
	Then $\abs{\grGR} < 18 \alpha n$.
\end{theorem}
\begin{proof}
  Combining the results of~\Cref{lemmaP,lemmaNP},
  $\abs{\grGR} = \abs{\grP} + \abs{\grNP} < 2 \alpha E(w) /\beta + 3 \alpha n / (1 - \beta)$
  for $2/3 \leq \beta < 1$.
  Applying~\Cref{lemmaExponent}, the term is upper bounded by $6 \alpha n / \beta + 3 \alpha n / (1 - \beta)$.
  The number is minimal for $\beta = 2/3$, yielding the bound $18 \alpha n$.
\end{proof}

We can bound the number of maximal $\alpha$-gapped palindromes by similar proofs to $28 \alpha n + 7 n$.
This bound solves an open problem in~\cite{gappedPalindroms}, where Kolpakov and Kucherov
conjectured that the number of $\alpha$-gapped palindromes with $\alpha \geq 2$ in a string is linear.
We briefly explain the main differences and similarities needed to understand the relationship between
gapped repeats and palindromes.
Let $\agr$ be a maximal $\alpha$-gapped repeat (or $\alpha$-gapped palindrome).
  If the gapped repeat (palindrome) has a periodic prefix~$\l\subarm$ generated by some run, the right arm has a periodic prefix (suffix)~$\r\subarm$
  generated by a run of the same period. 
  Since $\agr$ is maximal, both runs have to obey constraints that are similar in both cases, considering whether $\agr$ is a gapped repeat or a gapped palindrome.
  So it is easy to change the proof of \Cref{lemmaP} in order to work with palindromes.
  Like with aperiodic gapped repeats,
we can apply the point analysis to the aperiodic $\alpha$-gapped palindromes, too.
As main idea, we map a gapped palindrome  $\l\arm,\gap,\r\arm$ injectively to the pair of integers $(\e{\l\arm}, \abs{\gap})$, 
exchanging the period with the size of the gap. {See appendix for proofs.} %

\section{Algorithms}
The computational model we use to design and analyze our algorithms is the standard unit-cost RAM with logarithmic word size, which is generally used in the analysis of algorithms. 
In the upcoming algorithmic problems, we assume that the words we process are sequences of integers. In general, if the input word has length $n$ then we assume its letters are in $\{1,\ldots,n\}$, so each letter fits in a single memory-word. This is a common assumption in stringology (see, e.g., the discussion in \cite{KaSaBu06}).
For a word $w$, $|w|=n$, we build in $\bigo(n)$ time the suffix array as well as data structures allowing us to retrieve in constant time the length of the longest common prefix of any two suffixes $w[i,n]$ and $w[j,n]$
of $w$, denoted $\LCP_w(i,j)$ (the subscript $w$ is omitted when there is no danger of confusion). Such structures are called $\LCP$ data structures in the following (see, e.g., \cite{KaSaBu06,Gu97}). %
We begin with a simple lemma.
\begin{lemma}\label{periods}
Given a word $w$, $|w|=n$, we can process it in $\bigo(n)$ time such that, for each $i,p\leq n$, we can return in $\bigo(1)$ time the longest factor of period $p$ starting at position $i$ in $w$. 
\end{lemma}

Let $w$ be a word and $v$ be a factor of $w$ with $\per(v)=p$. 
Further, let $z$ be a subword of length $\ell |v|$ of $w$. 
An occurrence of $v$ in $z$ is a subword $(v,[i,i+|v|-1])$ of $z$; we say that $v$ occurs at position $i$ in $z$. 
For an easier presentation of our algorithm, we distinguish between two types of occurrences of $v$ in $z$.
On the one hand, we have the so-called \intWort{single occurrences}.
If $v$ is aperiodic, then all its occurrences in $z$ are single occurrences; there are $\bigo(\ell)$ such occurrences (see, e.g., \cite{KociumakaSPIRE2012}).
If $v$ is periodic, then a subword $(v,[i,i+|v|-1])$ of $z$ starting on position $i$ in $z$ is a single occurrence if $v$ occurs neither at position $i-p$ nor at position $i+p$ in $z$.
On the other hand, we have \intWort{occurrences of $v$ within a run} of $z$, whose period is $p=per(v)$.
That is, the subword  $(v,[i,i+|v|-1])$ starting on position $i$ in $z$ is an occurrence of $v$ within a run if $v$ occurs either at $i-p$ or at $i+p$.
We say that $(v,[i,i+|v|-1])$ is the \intWort{first occurrence} of $v$ in a run of period $p$ of $z$ if $v$ does not occur at $i-p$ but occurs at $i+p$.
Note that there are $\bigo(\ell)$ runs containing occurrences of $v$ in $z$, or, equivalently, $\bigo(\ell)$ first occurrences of $v$ in a run of period $p$.

Consequently, the occurrences of $v$ in $z$ can be succinctly represented as follows.
For the single occurrences we just store their starting position.
The occurrences of $v$ in a run~$\rep$ can be represented by the starting position of the first occurrence of $v$ in~$\rep$, 
together with the period of $v$,
since the starting positions of the occurrences of $v$ in~$\rep$ form an arithmetic progression of period $p$.

In our approach, basic factors (i.e., factors of length $2^k$, for $k\geq 1$) of the input word are important. For some integer $c\geq 2$, the occurrences of the basic factor $w[i,i+2^k-1]$ in a subword of length $c2^k$ can be represented in a compact manner: 
$\bigo(c)$ positions of the single occurrences of $w[i,i+2^k-1]$ and $\bigo(c)$ first occurrences of $w[i,i+2^k-1]$ in runs, together with the period of $w[i,i+2^k-1]$. 
We recall the next lemma (see \cite{FCT,KociumakaSPIRE2012}, appendix). 
\begin{lemma}\label{find_occ_range}
Given a word $w$ of length $n$ and an integer $c\geq 2$, we can process $w$ in time $\bigo(n\log n)$ such that given any basic factor $y=w[i,i+2^k-1]$ and any subword of $w$ $(z,[j,j+c 2^k -1])$, with $k\geq 0$, we can compute in $\bigo(\log \log n + c)$ time the representation of all the (single and within runs) occurrences of $y$ in $z$.
\end{lemma}

We now focus on short basic factors of words. The constant $16$ occurring in the following considerations can be replaced by any other constant; we just use it here so that we can apply these results directly in the main proofs of this section. 

Given a word $v$ and some integer $\beta\geq 16$ with $|v|=\beta \log n$,  as well as a basic factor $y=v[i2^k+1,(i+1)2^{k}]$, with $i,k\geq 0$ and $i2^k+1> (\beta-16)\log n$ (so occurring in the suffix of length $16 \log n$ of $v$), the occurrences of $y$ in $v$ can be represented as $\bigo(\beta)$ bit-sets, each containing $\bigo(\log n)$ bits, the $1$-bits marking the starting positions of the occurrences of $y$ in $v$. 
The next result can be shown using tools developed in~\cite{Gawrychowski11} (see also \cite{FCT} and the appendix). 
\begin{lemma}\label{find_occ_small}
Given a word $v$ and an integer $\beta> 16$, with $|v|=\beta \log n$, we can process $v$ in time $\bigo(\beta \log n)$ time such that given any basic factor $y=v[i2^k+1,(i+1)2^{k}]$ with $i,k\geq 0$ and $i2^k+1> (\beta-16)\log n$, we can find in $\bigo(\beta)$ time the $\bigo(\beta)$ bit-sets, each storing $\bigo(\log n)$ bits, characterizing all the occurrences of $y$ in $v$.
\end{lemma}

In the context of the previous lemma, once the occurrences of $y$ in $v$ are computed, given a subword $z$ of $v$ of length $|z|=c|y|$, 
for some $c\geq 1$, we can obtain in $\bigo(c)$ time both the single occurrences of $y$ in $z$ and the occurrences of $y$ within runs of $z$.
We just have to select (by bitwise operations on the bit-sets encoding the factors of $v$ that overlap $z$) the positions where $y$ occurs (so the positions of the $1$-bits in those bit-sets).
For each two consecutive such occurrences of $y$ we detect whether they are part of a run in $v$ and then skip over all the occurrences of $y$ from that run (and the corresponding parts of the bit-sets) before looking again for $1$-bits in the bit-sets; for the positions that form a run we store the first occurrence of $y$ and its period, while for the single occurrences we store the position of that occurrence.%

Now we can begin the presentation of the algorithm finding all the maximal $\alpha$-gapped repeats of a word.
We first show how to find maximal repeats with short arms.
\begin{lemma}\label{short_reps}
Given a word $w$ and $\alpha\geq 1$, we can find all the maximal $\alpha$-gapped repeats $u_\lambda, u', u_\rho$ occurring in $w$, with $|u_\rho|\leq 16 \log n$, in time $\bigo(\alpha n)$.
\end{lemma}
\begin{proof}
If a maximal $\alpha$-gapped repeat $u_\lambda, u', u_\rho$ (where we denote by $u$ the underlying factor of both arms), has $|u|\leq 16 \log n$, we get that $u_\rho$ must be completely contained in a subword $(w',[m\log n+1 , (m+17)\log n])$, for some $m$ with $\frac{n}{\log n}-17 \geq  m\geq 0$. By fixing the interval where $u_\rho$ may occur (that is, fix $m$), we also fix the place where $u_\lambda$ may occur. Indeed, %
the entire subword $u_\lambda, u', u_\rho $ is completely contained in the factor $x_m=(w'',[(m - 16 \alpha )\log n +1  ,  (m+17)\log n])$ (or, in a factor $x_m=(w'',[1, (m+17)\log n])$ if $(m - 16 \alpha )\log n +1<1$). 

Hence, we look for maximal $\alpha$-gapped repeats $u_\lambda, u',u_\rho$ completely contained in $x_m$ with $u_\rho$ completely contained in the suffix of length $16 \log n$ of $x_m$; then we repeat this process for all $m$.
To begin with, we process $x_m$ as in Lemma \ref{find_occ_small}, and construct $\LCP$-structures for it. 

Now, once we fixed the subword $x_m$ of $w$ where we search the maximal $\alpha$-gapped repeats, we try to fix also their length. That is, we find all maximal $\alpha$-gapped repeats $u_\lambda, u', u_\rho$ with $2^{k+1}\leq |u|\leq 2^{k+2}$ completely contained in $x_m$ with $u_\rho$ completely contained in the suffix of length $16 \log n$ of $x_m$; we execute this process for all $0\leq k\leq \log (16 \log n)$. Note that all the maximal $\alpha$-gapped repeats with arms shorter than $2$ (occurring anywhere in the word $w$) can be trivially found in $\bigo(\alpha n)$ time.

Since we can find occurrences of basic factors in $x_m$ efficiently, we try to build a maximal gapped repeat by extending 
gapped repeats whose arms contain a basic factor (see~\Cref{figExtendXm} in the appendix).
To this end, we analyze some \emph{subwords of} $x_m$:
If $2^{k+1}\leq |u|\leq 2^{k+2}$ then $u_\rho$ contains at least one subword $(y,[j2^k+1,(j+1)2^k])$ starting within its first $2^k$ positions.
A copy of the factor $y$ occurs also within the first $2^k$ positions of $u_\lambda$ (with the same offset with respect to the starting position of $u_\lambda$ as the offset of the occurrence of $y$ with respect to the starting position of $u_\rho$).
So, finding the respective copy of $y$ from $u_\lambda$ helps us discover the place where $u_\lambda$ actually occurs.
Indeed, assume that we identified the copy of $y$ from $u_\lambda$, and assume that this copy is $(y,[\ell+1,\ell+|y|])$; we try to build $u_\lambda$ and $u_\rho$ around these two occurrences of $y$, respectively.
Hence, in order to identify $u_\lambda$ and $u_\rho$ we compute the longest factor $p$ of $x_m$ that ends both on $j2^k$ and on $\ell$ and the longest factor $s$ that starts both on $(j+1)2^k+1$ and on $\ell+|y|+1$.
Now, if $\ell+|y|+|s|\leq j2^k-|p|$ then $u_\lambda$ is obtained by concatenating $p$ and $s$ around $x_m[\ell+1,\ell+|y|]$ while $u_\rho$ is obtained by concatenating $p$ and $s$ to the left and, respectively, right of $x_m[j2^k+1,(j+1)2^k]$; otherwise, the two occurrences of $y$ do not determine a maximal repeat.
Moreover, the repeat we determined is a valid solution of our problem only if its length is between $2^{k+1}$ and $2^{k+2}$, and its right arm contains position $j2^k+1$ of $x_m$ within its first $2^k$ positions.

Now we explain how to determine efficiently the copy of $y$ around which we try to build~$u_\lambda$. As $|u|<2^{k+2}$ and $|y|=2^k$ we get that the copy of $y$ that corresponds to $u_\lambda$ should be completely contained in the subword of $x_m$ of length $\alpha 2^{k+2}$ ending on position $j2^k$. As said above, we already processed $x_m$ to construct the data structures from Lemma \ref{find_occ_small}. Therefore, we can obtain in $\bigo(\alpha)$ time a representation of all the occurrences of $y$ inside the factor of length $\alpha 2^{k+2}$ ending on position $j2^k$. These occurrences can be single occurrences and occurrences within runs.
There are $\bigo(\alpha)$ single occurrences, and we can process each of them individually, as explained, to find the maximal $\alpha$-gapped repeat they determine together with the occurrence of $y$ from $u_\rho$. 
However, it is not efficient to do the same for the occurrences of $y$ within runs. For these (which are also $\bigo(\alpha)$ many) we proceed as follows.

Assume we have a run of occurrences of $y$ inside the factor of $x_m$ of length $\alpha 2^{k+2}$ ending on position $j2^k$. Let $\ell$ be the starting position of the first occurrence of $y$ in this run and let $p$ be the period of $y$. Now, using Lemma \ref{periods} we can determine the maximal $p$-periodic subword $r_\lambda$ of $x_m$ containing this run of $y$-occurrences. Similarly, we can determine the maximal $p$-periodic subword $r_\rho$ that contains the  occurrence of $y$ from $u_\rho$ (i.e., $x_m[j2^k+1,(j+1)2^k]$). To determine efficiently the $\alpha$-gapped repeats that contain $x_m[j2^k+1,(j+1)2^k]$ in the right arm and a corresponding occurrence of $y$ from $r_\lambda$ in the left arm we analyze several cases (see \Cref{figPeriodicExtension} in the appendix). 

Assume $u_\rho$ starts on a position of $r_\rho$, other than its first one.
Then $u_\lambda$ should also start on the first position of $r_\lambda$ (or we could extend both arms to the left, a contradiction to the maximality of the repeat).
If $u_\rho$ ends on a position to the right of $r_\rho$, then $u_\lambda$ also ends on a position to the right of $r_\lambda$, and, moreover, the suffix of $u_\lambda $ occurring after the end of $r_\lambda$ and the suffix of $u_\rho$ occurring after the end of $r_\rho$ are equal, and can be computed by a longest common prefix query on $x_m$.
This means that $u_\lambda$ can be determined exactly (we know where it starts and where it ends) so $u_\rho$ can also be determined exactly (we know where it ends), and we can check if the obtained repeat is indeed a maximal $\alpha$-gapped repeat, 
and the arms fulfill the required length conditions (i.e., length between $2^{k+1}$ and $2^{k+2}$, the right arm contains position $j2^k+1$ of $x_m$ within its first $2^k$ positions).
If $u_\rho$ ends exactly on the same position as $r_\rho$ then $u_\rho$ is periodic of period $p$; we just have to compute the longest $p$-periodic factor that ends on the same position as $r_\rho$ and starts at the same position as $r_\lambda$, and this can also be determined in constant time just by taking the longest $p$-periodic prefix of $r_\lambda$ which is also a suffix of $r_\rho$.
So, again, we can determine exactly $u_\lambda$ and $u_\rho$, and we can check if they form a maximal $\alpha$-gapped repeat with the arms fulfilling the length restrictions.
The final, and more complicated case, is when $u_\rho$ ends on a position of $r_\rho$, other than its last position.
In that case, we get that $u_\lambda=r_\lambda$ (or, otherwise, we could extend both arms to the right).
Essentially, this means that we know exactly where $u_\lambda$ is located and its length (and we continue only if this length is between $2^{k+1}$ and $2^{k+2}$); so $u_\lambda$ denotes a factor $z^h z'$ for some $z$ of length $p$.
Now, looking at the run $r_\rho$, we can get easily the position of the first occurrence of $z$ in that run, and the position of its last occurrence.
If the first occurrence is $\ell'$, then the occurrences of $z$ have their starting positions $\ell'$, $\ell'+p, \ldots, \ell'+tp $ for some $t$.
As we know the length of $u_\lambda$ and the fact that $u_\lambda, u', u_\rho$ is $\alpha$-gapped, we can determine in constant time the values $0\leq i\leq t$ such that $u_\rho$ may start on position $\ell'+ip$, the repeat we obtain is $\alpha$-gapped, and $u_\rho$ contains position $j2^k+1$ of $x_m$ within its first $2^k$ positions.
If $u_\rho$ is a prefix of $r_\rho$ we also have to check that we cannot extend simultaneously $u_\rho$ and $u_\lambda$ to the left; if $u_\rho$ is a suffix of $r_\lambda$ we have to check that we cannot extend simultaneously $u_\rho$ and $u_\lambda$ to the right.
Then we can return the maximal $\alpha$-gapped repeats we constructed.

The cases when $u_\rho$ starts on the first position of $r_\rho$ or when it starts on a position to the left of $r_\rho$ can be treated similarly{, and as efficiently (see Appendix).} 

This concludes our algorithm.
Its correctness follows from the explanations above.
Moreover, we can ensure that our algorithm finds and outputs each maximal repeat exactly once; this clearly holds when we analyze the repeats of $x_m$ for each $m$ separately.
 However, when moving from $x_m$ to $x_{m+1}$ we must also check that the right arm of each repeat we find is not completely contained in $x_m$ (so, already found).
This condition can be easily imposed in our search: when constructing the arms determined by a single occurrence of $y$, we check the containment condition separately; when constructing a repeat determined by a run of $y$-occurrences, we have to impose the condition that the right arm extends out of $x_m$ when searching the starting positions of the possible arms.

Next, we compute the complexity of the algorithm. Once we fix $m$, $k$, and $j$, our process takes $\bigo(\alpha + N_{j,m,k})$ time, where $N_{j,m,k}$ is the number of maximal $\alpha$-gapped repeats determined for the fixed $m,j,k$. So, the time complexity of the algorithm is:\\
 \centerline{$\bigo(n+\sum_{0\leq m\leq n/\log n}(16\alpha\log n+ \sum_{0\leq k\leq \log (16 \log n)} (\sum_{j\leq 16 \log n/ 2^k} (\alpha + N_{j,m,k})))) \subseteq \bigo(\alpha n)$,} 
 as the total number of maximal $\alpha$-gapped repeats is $\bigo(\alpha n)$ and we need $\bigo(|x_m|)$ preprocessing time for each $x_m$ and $\bigo(n)$ preprocessing time for $w$.
\end{proof}

Next, we find all maximal $\alpha$-gapped repeats with longer arms.
\begin{lemma}\label{long_reps}
Given a word $w$ and $\alpha\geq1$, we can find all the maximal $\alpha$-gapped repeats $u_\lambda, u', u_\rho$ occurring in $w$, with $|u_\rho|> 16 \log n$, in time $\bigo(\alpha n)$.
\end{lemma}
\begin{proof}
The general approach in proving this lemma is similar to that used in the proof of the previous result.
Essentially, when identifying a new maximal $\alpha$-gapped repeat, we try to fix the place and length of the right arm $u_\rho$ of the respective repeat, which restricts the place where the left arm $u_\lambda$ occurs.
This allows us to fix some long enough subword of $w$ as being part of the right arm, detect its occurrences that are possibly contained in the left arm, and, finally, to efficiently identify the actual repeat.
The main difference is that we cannot use the result of Lemma \ref{find_occ_small}, as we have to deal with repeats with arms longer than $16 \log n$.
Instead, we will use the structures constructed in Lemma \ref{find_occ_range}.
However, to get the stated complexity, we cannot apply this lemma directly on the word $w$, but rather on an encoded variant of $w$.

Thus, the first step of the algorithm is to construct a word $w'$, of length $\frac{n}{\log n}$, whose symbols, called \intWort{blocks}, encode $\log n $ consecutive symbols of $w$ grouped together.
That is, the first block of the new word corresponds to $w[1,\log n]$, the second one to $w[\log n+1,2\log n]$, and so on.
Hence, we have two versions of the word $w$: the original one, and the one where it is split in blocks.
It is not hard to see that the blocks can be encoded into numbers between~$1$ and $n$ in linear time.
Indeed, we build the suffix array and $\LCP$-data structures for $w$, and then we cluster together the suffixes of the suffix array that share a common prefix of length at least $\log n$.
Then, all the suffixes of a cluster are given the same number (between $1$ and $n$), and a block is given the number of the suffix starting with the respective block.

We can now construct in $\bigo(n)$ time the suffix arrays and $\LCP$-data structures for both $w$ and $w'$, as well as the data structures of Lemma \ref{find_occ_range} for the word $w'$. %

Now, we guess the length of the arms of the repeat.
We try to find the maximal $\alpha$-gapped repeats $u_\lambda, u', u_\rho$ of $w$ with $2^{k+1} \log n \leq |u_\lambda | \leq 2^{k+2}\log n$, $k\leq \log \frac{n}{\log n}-2$.
We fix $k$ and split again the word $w$, this time in factors of length $2^k \log n$, called \intWort{$k$-blocks}.
Assume that each split is exact (padding the word with some new symbols ensures this).

Now, if a maximal $\alpha$-gapped repeat  $u_\lambda, u', u_\rho$ with $2^{k+1} \log n \leq |u_\lambda | \leq 2^{k+2}\log n$ exists, then it contains an occurrence of a $k$-block within its first $2^k\log n$ positions.
So, let $z$ be a $k$-block and assume that it is the first $k$-block occurring in $u_\rho$ (in this way fixing a range where $u_\rho$ may occur).
Obviously, if $u_\rho$ contains $z$, then $u_\lambda $ also contains an occurrence of $z$; however, this occurrence is not necessarily starting on a position $j\log n + 1$ for some $j\geq 0$ (so, it is not necessarily a sequence of blocks).
But, at least one of the factors of length $2^{k-1}$ starting within the first $\log n$ positions of $z$ (which are not necessarily sequences of blocks) must correspond, in fact, to a sequence of blocks from the left arm  $u_\lambda$.
So, let us fix now a factor $y$ of length $2^{k-1}$ that starts within the first $\log n$ positions of $z$ (we try all of them in the algorithm, one by one).
As said, the respective occurrence of $y$ from $u_\rho$ is not necessarily a sequence of blocks (so it cannot be mapped directly to a factor of $w'$). But, we look for an occurrence of $y$ starting on one of the $ \alpha 2^{k+2} \log n$ positions to the left of $z$, corresponding to a sequence of blocks, and assume that the respective occurrence is exactly the occurrence of $y$ from $u_\lambda$.

By binary searching the suffix array of $w'$ (using $\LCP$-queries on $w$ to compare the factors of $\log n$ symbols of $y$ and the blocks of $w'$, at each step of the search) we try to detect a factor of $w'$ that encodes a word equal to $y$.
Assume that we can find such a sequence $y'$ of $2^{k-1}$ blocks of $w'$ (otherwise, $y$ cannot correspond to a sequence of blocks from $u_\lambda$, so we should try other factors of $z$ instead).
Using Lemma~\ref{find_occ_range} for $w'$, we get in $\bigo(\log\log |w'|+\alpha)$ time a representation of the occurrences of $y'$ in the range of $\alpha2^{k+2}$ blocks of~$w'$ occurring before the blocks of $z$; this range corresponds to an interval~of~$w$ with a length of $\alpha 2^{k+2} \log n$.

Further, we process these occurrences of $y'$ just like in the previous lemma.
Namely, the occurrences of $y'$ in that range are either single occurrences or occurrences within runs.
Looking at their corresponding factors from $w$, we note that each of these factors fixes a possible left arm $u_\lambda$; this arm, together with the corresponding arm $u_\rho$ can be constructed just like before.
In the case of single occurrences (which are at most $\bigo(\alpha)$, again), %
we try to extend both the respective occurrence and the occurrence of $y$ from $u_\rho$ both to the left and, respectively, to the right, simultaneously, and see if we can obtain in this way the arms of a valid maximal $\alpha$-gapped repeat.
Note that we must check also that the length of the arm of the repeat is between $2^{k+1}$ and $2^{k+2}$, and that $z$ is the first $k$-block of the right arm.
As before, complications occur when the occurrences of $y'$ are within runs.
In this case, the run of occurrences of $y'$ does not necessarily give us the period of $y$, but a multiple of this period that can be expressed also as a multiple of $\log n$ (or, in other words, the minimum period of $y$ is a multiple of the block-length).
This, however, does not cause any problems, as the factor $y$ from $u_\rho$ should always correspond to a block sequence from $u_\lambda$, so definitely to one of the factors encoded in the run of occurrences of $y'$.

Therefore, by determining the maximal factor that contains $y$ and has the same period as the run of occurrences of $y'$ (with the period measured in $w$), we can perform a very similar analysis to the corresponding one from the case when we searched maximal $\alpha$-gapped repeats with arms shorter than $16 \log n$. 

It remains to prove that each maximal gapped repeat is counted only once. 
Essentially, the reason for this is that for two separate factors $y_1$ and $y_2$ (of length $2^{k-1}$) occurring in the first $\log n$ symbols of $z$ we cannot get occurrences of the corresponding factors $y'_1$ and $y'_2$ that define the same repeat; in that case, the distance between $y'_1$ and $y'_2$ should be at least one block, so the distance between $y_1$ and $y_2$ should be at least $\log n$, a contradiction. Similarly, if we have a factor $y$ occurring in the first $\log n$ symbols of some $k$-block $z_1$ such that this factor determines an $\alpha$-gapped maximal repeat, then the same maximal repeat cannot be determined by a factor of another $k$-block, since $z_1$ is the first $k$-block of $u_\rho$.

The correctness of the algorithm described above follows easily from the explanations given in the proofs of the last two lemmas.
Let us evaluate its complexity.
The preprocessing phase (construction of $w'$ and of all the needed data structures) takes $\bigo(n)$ time.
Further, we can choose $k$ (and implicitly an interval for the length of the arms of the repeats) such that $k\leq \log \frac{n}{\log n}-2$.
After choosing $k$, we can choose a $k$-block $z$ in $\frac{n}{2^k \log n}$ ways.
Further, we analyze each factor $y$ of length $2^{k-1}$ starting within the first $\log n$ positions of the chosen $k$-block $z$.
For each such factor $y$ we find in $\bigo(\log \frac{n}{\log n} + \log\log n+ \alpha)$ time the representation of the occurrences of the block encoding the occurrence of $y$ from $u_\lambda$.
From each of the $\bigo(\alpha)$ single occurrences we check whether it is possible to construct a maximal $\alpha$-gapped repeat in $\bigo(1)$ time.
We also have $\bigo(\alpha)$ occurrences of the block encoding $y$ in runs, and each of them is processed in $\bigo(N_{z,y})$ time, where $N_{z,y}$ is the number of maximal $\alpha$-gapped repeats we find for some $z$ and $y$.
Overall, this adds up to a total time of $\bigo(n \log n + \alpha n)$, as the total number of maximal $\alpha$-gapped repeats in $w$ is upper bounded by $\bigo(\alpha n)$.
If $\alpha\geq \log n$, the statement of the lemma follows.
If $\alpha < \log n$ we proceed as follows.

Initially, we run the algorithm only for $k> \log \log n$ and find the maximal $\alpha$-gapped repeats $u_\lambda u' u_\rho$ with $2^{\log\log n} \log n \leq |u_\lambda |$, in $\bigo(\alpha n)$ time.
Further, we search maximal $\alpha$-gapped repeats with shorter arms.
Now, $|u_\lambda|$ is upper bounded by $2^{\log\log n +1}\log n=2 (\log n)^2$, so $|u_\lambda u'u_\rho |\leq \ell_0$, for $\ell_0=\alpha\cdot 2 (\log n)^2+2 (\log n)^2=2(\alpha+1)(\log n)^2$.
Such an $\alpha$-gapped repeat $u_\lambda u'u_\rho $ is, thus, contained in (at least) one factor of length $ 2\ell_0$ of $w$, starting on a position of the form $1+m\ell_0$ for $m\geq 0$.
 So, we take the factors $w[1+m\ell_0,(m+2)\ell_0]$ of $w$, for $m\geq 0$, and apply for each such factor, separately, the same strategy as above to detect the maximal $\alpha$-gapped repeats contained completely in each of them.
The total time needed to do that is {$\bigo\left(\alpha\ell_0 \frac{n}{\ell_0} + N_{\ell_0}\right)=\bigo(\alpha n)$, where $N_{\ell_0}$ is the number of repeats we find; moreover, we can easily ensure that a maximal repeat is not output twice (that is, ensure always that the gapped repeats we produce were not already contained in a previously processed interval).
	Hence, we find all maximal $\alpha$-gapped repeats $u_\lambda u'u_\rho $ with $2^{\log \log (2\ell_0)} \log (2\ell_0) \leq |u| $.
This means we find all the maximal $\alpha$-gapped repeats with $|u|\geq 2^{\log \log (2\ell_0) +1} \log (2\ell_0)$.
Since $2^{\log \log (2\ell_0) +1} \log (2\ell_0) \leq 16\log n$ (for $n$ large enough, as $\alpha\leq \log n$), we can apply Lemma~\ref{short_reps} for gapped repeats with an arm-length smaller than
$2^{\log \log (2\ell_0) +1} \log (2\ell_0)$.
}
\end{proof}

Putting together the results of Lemmas \ref{short_reps} and \ref{long_reps} we get the following theorem. 
\begin{theorem}\label{total}
Given a word $w$ and $\alpha\geq 1$, we can compute $\grGR$  in time $\bigo(\alpha n)$.
\end{theorem}

{By a completely similar approach we can compute $\gpGR$, generalizing the algorithm of~\cite{gappedPalindroms}.
To this end, we construct $\LCP$-structures for $w\pali{w}$ (allowing us to test efficiently whether a factor $\pali{w[i,j]}$ occurs at some position in $w$).
When we search the $\alpha$-gapped palindromes $u_\lambda,v,u_\rho$ (with $\r\arm \equiv \pali{\l\arm}$), we split again $w$ in blocks and $k$-blocks, for each $k\leq \log |w|$, to check whether there exists such an $u_\lambda,v,u_\rho$ with $2^{k}\leq |u_\lambda |\leq 2^{k+1}$. 
This search is conducted pretty much as in the case of repeats, only that now when we fix some factor $y$ of $u_\rho$, we have to look for the occurrences of $\pali{y}$ in the factor of length $\bigo(\alpha |u_\rho|)$ preceding it; the $\LCP$-structures for $w\pali{w}$ are useful for this, because, as explained above, they allow us to efficiently search the mirror images of factors of $w$ inside $w$.
Thus, given a word $w$ and $\alpha\geq 1$, we can compute $\gpGR$  in time $\bigo(\alpha n)$.}

\newpage
\bibliography{f_periodic}

\begin{thebibliography}{10}

\bibitem{runstheorem}
Hideo Bannai, Tomohiro I, Shunsuke Inenaga, Yuto Nakashima, Masayuki Takeda,
  and Kazuya Tsuruta.
\newblock {The Runs Theorem}.
\newblock {\em CoRR}, abs/1406.0263, 2014.

\bibitem{Brodal}
Gerth~St{\o}lting Brodal, Rune~B. Lyngs{\o}, Christian N.~S. Pedersen, and Jens
  Stoye.
\newblock Finding maximal pairs with bounded gap.
\newblock In {\em Proc. 10th Annual Symposium on Combinatorial Pattern
  Matching}, volume 1645 of {\em LNCS}, pages 134--149. Springer, 1999.

\bibitem{power_of_SA}
Maxime Crochemore, Costas~S. Iliopoulos, Marcin Kubica, Wojciech Rytter, and
  Tomasz Walen.
\newblock Efficient algorithms for two extensions of {LPF} table: The power of
  suffix arrays.
\newblock In {\em Proc. {SOFSEM} 2010}, volume 5901 of {\em LNCS}, pages
  296--307, 2010.

\bibitem{KKC15}
Maxime Crochemore, Roman Kolpakov, and Gregory Kucherov.
\newblock {Optimal searching of gapped repeats in a word}.
\newblock {\em ArXiv e-prints 1309.4055}, 2015.

\bibitem{Cro2011}
Maxime Crochemore and German Tischler.
\newblock Computing longest previous non-overlapping factors.
\newblock {\em Inf. Process. Lett.}, 111(6):291--295, February 2011.

\bibitem{MFCS}
Marius Dumitran and Florin Manea.
\newblock Longest gapped repeats and palindromes.
\newblock In {\em Proc. MFCS 2015}, volume 9234 of {\em LNCS}, pages 205--217.
  Springer, 2015.

\bibitem{Gawrychowski11}
Pawel Gawrychowski.
\newblock Pattern matching in {Lempel-Ziv} compressed strings: Fast, simple,
  and deterministic.
\newblock In {\em Proc. ESA}, volume 6942 of {\em LNCS}, pages 421--432, 2011.

\bibitem{FCT}
Pawel Gawrychowski and Florin Manea.
\newblock Longest $\alpha$-gapped repeat and palindrome.
\newblock In {\em Proc. FCT 2015}, volume 9210 of {\em LNCS}, pages 27--40.
  Springer, 2015.

\bibitem{Gu97}
Dan Gusfield.
\newblock {\em Algorithms on strings, trees, and sequences: computer science
  and computational biology}.
\newblock Cambridge University Press, New York, NY, USA, 1997.

\bibitem{KaSaBu06}
Juha K\"{a}rkk\"{a}inen, Peter Sanders, and Stefan Burkhardt.
\newblock Linear work suffix array construction.
\newblock {\em J. ACM}, 53:918--936, 2006.

\bibitem{KociumakaSPIRE2012}
Tomasz Kociumaka, Jakub Radoszewski, Wojciech Rytter, and Tomasz Walen.
\newblock Efficient data structures for the factor periodicity problem.
\newblock In {\em Proc. {SPIRE}}, volume 7608 of {\em LNCS}, pages 284--294,
  2012.

\bibitem{KK_SPIRE}
Roman Kolpakov and Gregory Kucherov.
\newblock Finding repeats with fixed gap.
\newblock In {\em {Proc. SPIRE}}, pages 162--168, 2000.

\bibitem{gappedPalindroms}
Roman Kolpakov and Gregory Kucherov.
\newblock Searching for gapped palindromes.
\newblock {\em Theoretical Computer Science}, 410(51):5365 -- 5373, 2009.
\newblock Combinatorial Pattern Matching.

\bibitem{KolpakovPPK14}
Roman Kolpakov, Mikhail Podolskiy, Mikhail Posypkin, and Nickolay Khrapov.
\newblock Searching of gapped repeats and subrepetitions in a word.
\newblock In {\em Proc. {CPM}}, volume 8486 of {\em LNCS}, pages 212--221,
  2014.

\bibitem{TanimuraFIIBT15}
Yuka Tanimura, Yuta Fujishige, Tomohiro I, Shunsuke Inenaga, Hideo Bannai, and
  Masayuki Takeda.
\newblock A faster algorithm for computing maximal $\alpha$-gapped repeats in a
  string.
\newblock In {\em Proc. {SPIRE} 2015}, volume 9309 of {\em LNCS}, pages
  124--136. Springer, 2015.

\bibitem{Boas75}
Peter van Emde~Boas.
\newblock Preserving order in a forest in less than logarithmic time.
\newblock In {\em Proc. FOCS}, pages 75--84, 1975.

\end{thebibliography}

\newpage
\pagestyle{empty}
\section*{Appendix}
{%
\section{A short recent history of this problem}

The problem of searching for gapped repeats and palindromes in words is not so new (see \cite{Gu97,Brodal,KK_SPIRE}), and different solutions were proposed depending on the type of restrictions imposed on the gap.
In \cite{gappedPalindroms}, Kolpakov and Kucherov introduced the notion of $\alpha$-gapped palindromes, and showed how to compute
the set $\gpGR$ of all maximal $\alpha$-gapped palindromes in $\bigo(\alpha^2 n + |\gpGR|)$ time for any input word~$w$ of length $n$ over constant alphabets.
In \cite{KolpakovPPK14}, Kolpakov et al. introduced the notion of $\alpha$-gapped repeats, and showed that the number of maximal $\alpha$-gapped repeats in a word~$w$ of length $n$ is $\bigo(\alpha^2 n)$, and that this maximal set $\grGR$ can  be computed in $\bigo(\alpha^2 n)$ time.
They suggested two open problems concerning this:
\begin{itemize}
	\item closing the gap between the upper bound $\bigo(\alpha^2n)$ and the lower bound $\Omega(\alpha n)$ 
		for the number of maximal $\alpha$-gapped repeats, and 
	\item developing a more efficient algorithm.
\end{itemize}
Their preliminary experiments suggested that the upper bound should be $\bigo(\alpha n)$.

In February 2015, at the Stringmasters meeting in Warsaw, the third author of this paper pointed that obtaining a bound on the number of $\alpha$-gapped repeats seems to be an interesting open problem, and also conjectured that this bound is $\bigo(\alpha n)$.
As a reaction, several connected papers appeared: \cite{FCT} shows how to compute the longest $\alpha$-gapped repeat/palindrome in $\bigo(\alpha n)$ time, \cite{MFCS} shows how to compute a series of data structures that can give the longest $2$-gapped repeat/palindrome that starts at each position (and the results generalize easily to arbitrary $\alpha$), \cite{TanimuraFIIBT15} gives an $\bigo(\alpha n +|\grGR|)$-time solution to find all maximal $\alpha$-gapped repeats for an input word over constant alphabets.

Finally, in August 2015, the fourth author of this paper announced on the Stringmasters web-page\footnote{\url{http://stringmasters.mimuw.edu.pl/open\_problems.html}} that the bound on the number of $\alpha$-gapped repeats is indeed $\bigo(\alpha n)$; together with \cite{TanimuraFIIBT15} this lead to an optimal algorithm for solving the problem of finding all $\alpha$-gapped repeats in the particular case of constant alphabets.
This announcement was followed by the paper \cite{KKC15} which basically confirmed the bound $\bigo(\alpha n)$ and also gave an algorithm computing efficiently all $\alpha$-gapped maximal repeats for constant alphabets.
Our paper concludes in big measure this line of research: we give concrete bounds on the number of $\alpha$-gapped repeats and $\alpha$-gapped palindromes, and, building on the approach from \cite{FCT}, we give optimal algorithms for finding them in the usual case of integer alphabets.

}

\section{Gapped repeats}
The following figure might ease the understanding the proof of \Cref{lemmaP} regarding periodic gapped repeats.
\begin{figure}[h]
	\centering{%
		\begin{tikzpicture}[yscale=0.4,xscale=0.5]
  \col[arm]{2}{$\l\arm$}
  \void{1}
  \col[arm]{2}{$\r\arm$}
  \newrow
  \col[subarm]{1}{$\l\subarm$}
  \void{2}
  \col[subarm]{1}{$\r\subarm$}
  \newrow
  \col[rep]{1}{$\l\rep$}
  \void{1.5}
  \col[rep]{1.5}{$\r\rep$}
	\end{tikzpicture}
	\hspace{2em}
		\begin{tikzpicture}[yscale=0.4,xscale=0.5]
  \col[arm]{2}{$\l\arm$}
  \void{1}
  \col[arm]{2}{$\r\arm$}
  \newrow
  \col[subarm]{1}{$\l\subarm$}
  \void{2}
  \col[subarm]{1}{$\r\subarm$}
  \newrow
  \void{-0.5}
  \col[rep]{1.5}{$\l\rep$}
  \void{2}
  \col[rep]{1}{$\r\rep$}
	\end{tikzpicture}
}
\caption{%
  The equation~$\b{\l\arm} = \b{\l\rep}$ or $\b{\r\arm} = \b{\r\rep}$ must hold in the setting of \Cref{lemmaP}.
  By the maximality property of runs, $\e{\l\rep} = \e{\l\subarm}$ and $\e{\r\rep} = \e{\r\subarm}$.
  }
\label{figPeriodic}
\end{figure}

\section{Gapped palindromes}

\begin{Lemma}\label{lemmaPPali}
	Given a word~$w$, and two real numbers $\alpha > 1$ and $0 < \beta < 1$. Then
		$\abs{\gpP}$ is at most $2 (\alpha+1) E(w) / \beta$. 
\end{Lemma}
\begin{proof}
  Let $\agr \in \grP$ (or $\agr \in \gpP$).
  The gapped palindrome $\agr$ consists of the triple $(\l\arm,\gap,\r\arm) \in \grP$.
  By definition, the left arm~$\l\arm$ has a periodic prefix $\l\subarm$ of length at least $\beta \abs{\l\arm}$.
  Let $\l\rep$ denote the run that generates $\l\subarm$, i.e.,
  $\l\subarm \subseteq \l\rep$ and they both have the common shortest period $\period$.
  By the definition of gapped palindromes, there is a right copy $\r\subarm$ of $\l\subarm$ contained in $\r\arm$ with
  \[
	  \r\subarm = 
		  w[\b{\r\arm}+(\e{\l\arm}-\e{\l\subarm}), \b{\r\arm}+(\e{\l\arm}-\e{\l\subarm})+\abs{\l\subarm}-1] \equiv \pali{\l\subarm}
  \]
  Let $\r\rep$ be a run generating $\r\subarm$ (it is possible that $\r\rep$ and $\l\rep$ are identical).
  By definition, $\r\rep$ has the same period~$\period$ as $\l\rep$.
  In the following, we will see that $\agr$ is uniquely determined by $\l\rep$ and 
 the distance $d := \b{\r\subarm} - \e{\l\subarm}$, if $\agr$ is a periodic gapped palindrome.
  We will fix $\l\rep$ and study how many maximal periodic gapped repeats (palindromes) can be generated by $\l\rep$.

\begin{figure}[h]
	\centering{%
			\begin{tikzpicture}[yscale=0.4,xscale=0.5]
			\col[arm]{2}{$\l\arm$}
			\void{1}
			\col[arm]{2}{$\r\arm$}
			\newrow
			\col[subarm]{1}{$\l\subarm$}
			\void{3}
			\col[subarm]{1}{$\r\subarm$}
			\newrow
			\col[rep]{1}{$\l\rep$}
			\void{3}
			\col[rep]{1.5}{$\r\rep$}
			\end{tikzpicture}
			\hspace{2em}
\begin{tikzpicture}[yscale=0.4,xscale=0.5]
\col[arm]{2}{$\l\arm$}
\void{1}
\col[arm]{2}{$\r\arm$}
\newrow
\col[subarm]{1}{$\l\subarm$}
\void{3}
\col[subarm]{1}{$\r\subarm$}
\newrow
\void{-0.5}
\col[rep]{1.5}{$\l\rep$}
\void{3}
\col[rep]{1}{$\r\rep$}
\end{tikzpicture}
	}
	\caption{%
  The equation~$\b{\l\arm} = \b{\l\rep}$ or $\e{\r\arm} = \e{\r\rep}$ must hold in the setting of \Cref{lemmaPPali}.
  By the maximality property of runs, $\e{\l\rep} = \e{\l\subarm}$ and $\b{\r\rep} = \b{\r\subarm}$.
  }
	\label{figPeriodicPal}
\end{figure}

  Since $\agr$~is maximal, $\b{\l\arm} = \b{\l\rep}$ or $\e{\r\arm} = \e{\r\rep}$ must hold;
  otherwise we could extend $\agr$ outwards.
  We analyze the case $\b{\l\subarm} = \b{\l\rep}$, the other is treated exactly in the same way by symmetry.
    The gapped palindrome~$\agr$ is identified by~$\l\rep$ and~$d$.
    We fix $\l\rep$ and count the number of possible values for $d$.
    Given two different periodic $\alpha$-gapped palindromes with the distances $d_1$ and $d_2$, 
    the difference between $d_1$ and $d_2$ must be at least $\period$, due to~\Cref{lemmaRepetitionsOverlap}.
    It follows from $\abs{\l\arm} \leq \abs{\l\subarm} / \beta$ that
    $d = \abs{\gap} + 2(\abs{\l\arm}-\abs{\l\subarm}) \leq \abs{\gap} + 2 \abs{\l\subarm} / \beta$.
    Since $\agr$ is $\alpha$-gapped, $\abs{\gap} \leq (\alpha - 1) \abs{\l\arm} \leq (\alpha - 1) \abs{\l\subarm} / \beta$, and hence,
    $1 \leq d \leq \abs{\l\subarm} (\alpha+1) / \beta$.
    Then the number of possible values for the distance~$d$ is bounded by 
    $\abs{\l\subarm} (\alpha+1) / (\beta \period) \leq \abs{\l\rep} (\alpha+1) / (\beta \period) = \exp(\l\rep) (\alpha+1) / \beta$.
    In total, the number of maximal $\alpha$-gapped palindromes in this case is bounded by $(\alpha+1) E(w) / \beta$ for the case
	$\b{\l\arm} = \b{\l\rep}$.
  Summing up we get the bound $2 (\alpha+1) E(w) / \beta$.
  \end{proof}

\begin{Lemma}\label{lemmaPalNPCover}
	Given a word~$w$, and two real numbers $\alpha > 1$ and $6/7 \leq \beta < 1$.
	The points mapped by two different maximal gapped palindromes in $\gpNP$ cannot $\frac{1-\beta}{\alpha}$-cover the same point.
\end{Lemma}
\begin{proof}
	Let $\agr = \l\arm,\gap,\r\arm$ and $\s\agr = \s{\l\arm},\s{\gap},\s{\r\arm}$ be two different gapped palindromes in $\gpNP$.
  Set $\arm := \abs{\l\arm} = \abs{\r\arm}$, $\s\arm := \abs{\s{\l\arm}} = \abs{\s{\r\arm}}$, $\gapS := \abs{\gap}$ and $\s{\gapS} := \abs{\s{\gap}}$.
  We map the maximal gapped palindromes~$\agr$ and~$\s\agr$ to the points~$(\e{\l\arm}, \gapS)$ and~$(\e{\s{\l\arm}},\s{\gapS})$, respectively.
  Assume, for the sake of contradiction,  that both points $\frac{1-\beta}{\alpha}$-cover the same point $(x, y)$.

  Let $\zarm := \abs{\e{\l\arm} - \e{\s{\l\arm}}}$ be the difference of the endings of both left arms, and
  $\l\subarm := w[ [\b{\l\arm},\e{\l\arm}]\cap [\b{\s{\l\arm}},\e{\s{\l\arm}}] ]$ 
  be the overlap of $\l\arm$ and $\s{\l\arm}$.
  Let $\subarm = \abs{\l\subarm}$, 
  and let $\r\subarm$ (resp. $\s{\r\subarm}$) be the reversed copy of $\l\subarm$ based on $\agr$ (resp. $\s\agr$).

  {\bf Sub-Claim:}  The overlap~$\l\subarm$ is not empty, and $\r\subarm \not= \s{\r\subarm}$

  {\bf Sub-Proof.}
  Assume for this sub-proof that $\e{\l\arm} < \e{\s{\l\arm}}$ (otherwise exchange $\agr$ with $\s\agr$, or yield the contradiction $\agr = \s\agr$). 
  By combining the $(1-\beta)/\alpha$-cover property with the fact that $\s\agr$ is $\alpha$-gapped, we yield
  \(
	  \e{\s{\l\arm}} - \s\arm \le \e{\s{\l\arm}} - \s{\gapS}(1-\beta)/\alpha \le x \le \e{\l\arm} < \e{\s{\l\arm}}.
  \)
  So the subword $w[\e{\l\arm}]$ is contained in $\s{\l\arm}$.
  If $\r\subarm = \s{\r\subarm}$, then we get a contradiction to the maximality of $\agr$:
  By the above inequality, $w[\e{\l\arm}+1]$ is contained in $\s{\l\arm}$, too. 
  Since $\s\agr$ is a gapped palindrome, the character~$w[\e{\l\arm}+1]$ occurs in $\s{\r\arm}$, exactly at $w[\b{\r\arm}-1]$.
  \qed{}

  \Wlogg{} let $\gapS \le \s{\gapS}$.
  Then
  \begin{equation}\label{equOrdGapRuleInvPal}
    \s{\gapS} - \frac{\s{\gapS}(1-\beta)}{\alpha} \le y \le \gapS \le \s{\gapS}.
  \end{equation}
  So the difference of both gaps is
  \begin{equation}\label{equOrdGapRulePal}
    0 \le \delta := \s{\gapS} - \gapS \le \s{\gapS}(1-\beta)/\alpha \le \s\arm(1-\beta).
  \end{equation}

  By case analysis, we show that $\l\subarm$ is periodic,
  which leads to the contradiction that $\agr$ or $\s\agr$ are in $\gpP$.

  {\bf 1. Case $\e{\l\arm} \le \e{\s{\l\arm}}$}.
  Since $\e{\s{\l\arm}} - \s{\gapS}(1-\beta)/\alpha \le x \le \e{\l\arm} \le \e{\s{\l\arm}}$,
  \begin{equation}\label{equOrdGapRulePosFirstLeftPal}
    \zarm = \e{\s{\l\arm}} - \e{\l\arm} \le \s{\gapS}(1-\beta)/\alpha \le \s\arm(1-\beta).
  \end{equation}
  Since $\l\subarm$ is a prefix of $\s{\l\arm}$ and a suffix of $\l\arm$, 
  the reverse copy~$\r\subarm$ is a suffix of $\s{\r\arm}$ and a prefix of $\r\arm$.
  The starting positions of both right copies $\s{\r\subarm}$ and $\r\subarm$ differ by $\b{\s{\r\subarm}} - \b{\r\subarm} = 2\zarm + \delta > 0$.
  By~\Cref{equOrdGapRulePal,equOrdGapRulePosFirstLeftPal}, we get $2\zarm + \delta \le 3 \s{\gapS}(1-\beta)/\alpha \le 3 \s\arm(1-\beta)$.

  \begin{figure}[h]
    \centering{%
		\begin{tikzpicture}[yscale=0.4,xscale=0.92]
  \col[arm]{3.5}{$\l\arm$}
  \col[gap]{2}{$\gap$}
  \col[arm]{3.5}{$\r\arm$}
  \newrow
  \void{0.5}
  \col[arm]{3.5}{$\s{\l\arm}$}
  \col[gap]{2.5}{$\s\gap$}
  \col[arm]{3.5}{$\s{\r\arm}$}
  \newrow
  \void{0.5}
  \col[subarm]{3}{$\l\subarm$}
  \col[diff]{0.5}{$\zarm$}
  \void{1.5}
  \col[subarm]{3}{$\r\subarm$}
  \newrow
  \void{5.5}
  \col[rep]{1.5}{$2\zarm+\delta$}
  \col[subarm]{3}{$\s{\r\subarm}$}
	\end{tikzpicture}
    }
    \caption{Sub-Case 1a}
  \end{figure}

  {\bf 1a. Sub-Case $\b{\l\arm} \le \b{\s{\l\arm}}$}.
  By~\Cref{equOrdGapRulePosFirstLeftPal}, we get $\subarm = \s\arm - \zarm \ge \s\arm \beta$.
  It follows from $6/7 \leq \beta < 1$ that $\subarm / (2\zarm + \delta) \ge \s\arm \beta / 3 \s\arm (1-\beta) = \beta / (3(1-\beta)) \ge 2$,
  which means that $\r\subarm$ and $\s{\r\subarm}$ overlap by at least half of their common length, and $\l\subarm$ is periodic.
  Since $\l\subarm$ is a prefix of $\s{\l\arm}$ of length $\subarm \ge \s\arm \beta$, $\s\agr$ is in $\gpP$, a contradiction.

  \begin{figure}[h]
    \centering{%
		\begin{tikzpicture}[yscale=0.4,xscale=0.92]
  \void{0.5}
  \col[arm]{3.5}{$\l\arm$}
  \col[gap]{2}{$\gap$}
  \col[arm]{3.5}{$\r\arm$}
  \newrow
  \col[arm]{4.5}{$\s{\l\arm}$}
  \col[gap]{2.5}{$\s\gap$}
  \col[arm]{4.5}{$\s{\r\arm}$}
  \newrow
  \void{0.5}
  \col[subarm]{3.5}{$\l\subarm$}
  \col[diff]{0.5}{$\zarm$}
  \void{1.5}
  \col[subarm]{3.5}{$\r\subarm$}
  \newrow
  \void{6}
  \col[rep]{1.5}{$2\zarm+\delta$}
  \col[subarm]{3.5}{$\s{\r\subarm}$}
	\end{tikzpicture}
    }
    \caption{Sub-Case 1b}
  \end{figure}

  {\bf 1b. Sub-Case $\b{\l\arm} > \b{\s{\l\arm}}$}.
  We conclude that $\l\subarm = \l\arm$.
  By~\Cref{equOrdGapRuleInvPal},
  \begin{equation}\label{equOrdGapRuleTakeFirstArmPal}
  	\arm \ge \gapS/\alpha \ge \frac{\s{\gapS}}{\alpha} (1 - \frac{1-\beta}{\alpha}) \ge \s{\gapS}\beta/\alpha.
  \end{equation}
  It follows from $6/7 \leq \beta < 1$ that $\subarm / (2\zarm + \delta) \ge \s{\gapS} \alpha \beta / (3 \alpha \s{\gapS} (1-\beta)) = \beta / (3(1-\beta)) \ge 2$,
  which means that $\l\subarm = \l\arm$ is periodic.
  Hence $\agr$ is in $\gpP$, a contradiction.

  {\bf 2. Case $\e{\l\arm} > \e{\s{\l\arm}}$}.
  Since $\e{\l\arm} - \gapS(1-\beta)/\alpha \le x \le \e{\s{\l\arm}} \le \e{\l\arm}$,
  \begin{equation}\label{equOrdGapRulePosFirstRightPal}
    \zarm = \e{\l\arm} - \e{\s{\l\arm}} \le \gapS(1-\beta)/\alpha \le \s\gapS(1-\beta)/\alpha \le \s\arm(1-\beta).
  \end{equation}
  The starting positions of both right copies differ by $\abs{\b{\r\subarm} - \b{\s{\r\subarm}}} = \abs{2\zarm - \delta}$.
  Since $2\zarm - \delta \le \max\tuple{\delta,2\zarm}$,
  we get $\abs{2\zarm - \delta} \le 2 \s\gapS (1-\beta)/\alpha \le 2 \s\arm(1-\beta)$ by~\Cref{equOrdGapRulePal,equOrdGapRulePosFirstRightPal}.

  \begin{figure}[h]
    \centering{%
		\begin{tikzpicture}[yscale=0.4,xscale=0.92]
  \col[arm]{5}{$\l\arm$}
  \col[gap]{2.5}{$\gap$}
  \col[arm]{5}{$\r\arm$}
  \newrow
  \void{0.5}
  \col[arm]{3.5}{$\s{\l\arm}$}
  \col[gap]{3}{$\s\gap$}
  \col[arm]{3.5}{$\s{\r\arm}$}
  \newrow
  \void{0.5}
  \col[subarm]{3.5}{$\l\subarm$}
  \col[diff]{1}{$\zarm$}
  \void{2}
  \col[rep]{1.5}{${2\zarm-\delta}$}
  \col[subarm]{3.5}{$\r\subarm$}
  \newrow
  \void{7}
  \col[subarm]{3.5}{$\s{\r\subarm}$}
	\end{tikzpicture}
    }
    \caption{Sub-Case 2a}
  \end{figure}
  {\bf 2a. Sub-Case $\b{\l\arm} \le \b{\s{\l\arm}}$}.
  We conclude that $\l\subarm = \s{\l\arm}$.
  It follows from $6/7 \leq \beta < 1$ that $\subarm / \abs{2\zarm - \delta} \ge \s\arm / (2 \s\arm (1-\beta)) = 1 / (2(1-\beta)) \ge 7/2 > 2$,
  which means that $\l\subarm = \s{\l\arm}$ is periodic.
  Hence $\s\agr$ is in $\gpP$, a contradiction.

  \begin{figure}[h]
    \centering{%
		\begin{tikzpicture}[yscale=0.4,xscale=0.92]
  \void{0.5}
  \col[arm]{4.5}{$\l\arm$}
  \col[gap]{2}{$\gap$}
  \col[arm]{4.5}{$\r\arm$}
  \newrow
  \col[arm]{4}{$\s{\l\arm}$}
  \col[gap]{2.5}{$\s\gap$}
  \col[arm]{4}{$\s{\r\arm}$}
  \newrow
  \void{0.5}
  \col[subarm]{3.5}{$\l\subarm$}
  \col[diff]{1}{$\zarm$}
  \void{1.5}
  \col[rep]{1.5}{${2\zarm-\delta}$}
  \col[subarm]{3.5}{$\r\subarm$}
  \newrow
  \void{6.5}
  \col[subarm]{3.5}{$\s{\r\subarm}$}
	\end{tikzpicture}
    }
    \caption{Sub-Case 2b}
  \end{figure}
  {\bf 2b. Sub-Case $\b{\l\arm} > \b{\s{\l\arm}}$}.
  By \Cref{equOrdGapRulePosFirstRightPal}, we get  $\zarm \le \arm(1-\beta)$ and thus $\subarm = \arm - \zarm \ge \beta \arm$.
  It follows from~\Cref{equOrdGapRuleTakeFirstArmPal,equOrdGapRulePal}, and $2(\sqrt{2}-1) < 6/7 \leq \beta < 1$ that 
  $\subarm / \abs{2\zarm - \delta} \ge \beta \arm / (2 \s\gapS (1-\beta)/\alpha) \ge \s\gapS \beta^2 / (2 \s\gapS (1-\beta)) 
  = \beta^2 / (2 (1-\beta)) > 2$,
  which means that $\l\subarm$ is periodic.
  Since $\l\subarm$ is a prefix of $\l\arm$ of length $\subarm \ge \arm \beta$, $\agr$ is in $\gpP$, a contradiction.
\end{proof}

The next lemma follows immediately from~\Cref{lemmaPoints,lemmaPalNPCover}.
\begin{Lemma}\label{lemmaPalNP}
	For $\alpha > 1$, $6/7 \leq \beta < 1$ and a word~$w$ of length $n$, $\abs{\gpNP} < 3 \alpha n / (1 - \beta)$.
\end{Lemma}

\begin{theorem}
	For $\alpha > 1$ and a word~$w$ of length $n$, $\abs{\gpGR} < 28 \alpha n + 7 n$.
\end{theorem}
\begin{proof}
  By~\Cref{lemmaP,lemmaPalNP},
  $\abs{\gpGR} = \abs{\gpP} + \abs{\gpNP} < 2 (\alpha+1) E(w) / \beta + 3 \alpha n / (1 - \beta)$ for every $6/7 \leq \beta < 1$.
  Applying~\Cref{lemmaExponent}, the term is upper bounded by $6 (\alpha+1) n / \beta + 3 \alpha n / (1 - \beta)$.
  This number is minimal when $\beta = 6/7$, yielding the bound $28 \alpha n + 7 n$.
\end{proof}

\section*{Auxiliary algorithmic results}

\noindent{\bf Proof of Lemma \ref{periods}}.
\begin{proof}
Once we produce $\LCP$ data structures for $w$, we just have to compute the longest common prefix of $w[i,n]$ and $w[i+p,n]$. If this prefix is $w[i+p,\ell]$, then $w[i,\ell]$ is the longest $p$-periodic factor starting on position $i$.
\end{proof}
\bigskip

\noindent{\bf Proof of Lemma \ref{find_occ_range}}
\begin{proof}
We construct the dictionary of basic factors of a word of length $n$ in $\bigo(n\log n)$ time and reorganise it such that for each basic factor we have an array with all its occurrences, ordered by their starting positions. For each such array we construct data structures that allow predecessor/successor search in $\bigo(\log\log n)$ time (see, e.g., \cite{Boas75}). When we have to return the occurrences of $y=w[i,i+2^k-1]$ in $z=w[j,j+c 2^k -1]$, we search in the basic factors-array corresponding to $w[i,i+2^k-1]$ the successor of $j$ and, respectively, the predecessor of $j+c2^k -1$ and then return a succinct representation of the occurrences of $w[i,i+2^k-1]$ between these two values. This representation can be obtained in $\bigo(c)$ time. We just have to detect those occurrences that form a run; this can be done with a constant number of $\LCP$ queries. Indeed, for two consecutive occurrences, we compute the length of their overlap, which gives us a period of $w[i,i+2^k-1]$. Then we look in $w$ to see how long the run with this period can be extended to the right, which gives us the number of occurrences of  $w[i,i+2^k-1]$ in that run. As their starting positions form an arithmetic progression, we can represent them compactly. So, we return the representation of the occurrences of $w[i,i+2^k-1]$ from this run, and then move directly to the first occurrence of $w[i,i+2^k-1]$ appearing after this run and still in the desired range; as there are at most $\bigo(c)$ runs and separate occurrences of the given basic factor that are in $z$, the conclusion follows.
\begin{figure}
	\centering{%
	\begin{tikzpicture}[yscale=0.4]
  \void{-0.5}
  \col[arm]{14}{$z = w[j,j+8 \cdot 2^j - 1]$}
  \newrow
  \col[rep]{3}{$\rep_1$}
  \void{1}
  \col[rep]{2.5}{$\rep_2$}
  \void{1}
  \col[rep]{2}{$\rep_3$}
  \newrow
  \col[subarm]{1.5}{$y$}
  \void{2.5}
  \col[subarm]{1.5}{$y$}
  \void{2}
  \col[subarm]{1.5}{$y$}
  \void{2.5}
  \col[subarm]{1.5}{$y$}
  \newrow
  \void{0.5}
  \col[subarm]{1.5}{$y$}
  \void{2.5}
  \col[subarm]{1.5}{$y$}
  \void{2}
  \col[subarm]{1.5}{$y$}
  \void{2.5}
  \newrow
  \void{1}
  \col[subarm]{1.5}{$y$}
  \void{2.5}
  \col[subarm]{1.5}{$y$}
  \void{2}
  \newrow
  \void{1.5}
  \col[subarm]{1.5}{$y$}
  \void{2.5}
  \newrow
  \lab{$i$}
  \void{1.5}
  \lab{$i+2^k$}
	\end{tikzpicture}
}
\caption{Occurrences of the basic factor $y=w[i,i+2^k-1]$ in $z=w[j,j+8\cdot2^k-1]$. The overlapping occurrences are part of runs, and they can be returned as such. The representation of the occurrences of $y$ in $z$ will return $4$ elements: $3$ runs and one separate occurrence.}
\end{figure}
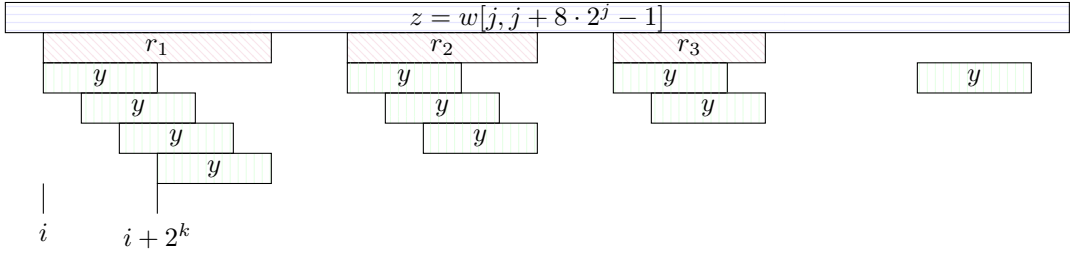
\end{proof}

\bigskip

\noindent{\bf Proof of Lemma \ref{find_occ_small}}.

In \cite{FCT} the following result was shown. Lemma \ref{find_occ_small} follows directly.

\begin{lemma}\label{find_occ_small_FCT}
Given a word $v$, $|v|=\alpha \log n$, we can process $v$ in time $\bigo(\alpha \log n)$ time such that given any basic factor $y=v[j\cdot2^k+1,(j+1)2^{k}]$ with $j,k\geq 0$ and $j2^k+1> (\alpha-1)\log n$, we can construct in $\bigo(\alpha)$ time $\bigo(\alpha)$ bit-sets, each storing $\bigo(\log n)$ bits, characterizing all the occurrences of $y$ in $v$.
\end{lemma}

The following lemma is useful for the proof. 
More precisely, one can also construct data structures that allow us fast identification of the suffixes of a word that start with a given basic factor.
\begin{lemma}[{\cite{Gawrychowski11}}]\label{find_node}
A word $w$ of length $n$ can be processed in $\bigo(n)$ time such that given any basic factor $w[i,i+2^k-1]$ with $k\geq 0$, we can retrieve in $\bigo(1)$ time the range of the suffix array of $w$ of suffixes starting with $w[i,i+2^k-1]$. Equivalently, we can find the node of the suffix tree of $w$ closest to the root, such that the label of the path from the root to that node has the prefix $w[i,i+2^k-1]$.
\end{lemma}

The proof of Lemma \ref{find_occ_small_FCT} follows.
\begin{proof}
We first show how the proof works for $\alpha=1$. 

We first build the suffix tree for $v$ in $\bigo(\log n)$ time. We further process this suffix tree such that we can find in constant time, for each factor $v[j\cdot 2^k+1,(j+1)2^{k}]$, the node~$u$ of the suffix tree which is closest to the root with the property that the label of the path from the root to~$u$  starts with $v[j\cdot 2^k+1,(j+1)2^{k}]$. According to Lemma \ref{find_node}, this can be done in linear time. 

Now, we augment the suffix tree in such a manner that for each node we store an additional bit-set, indicating the positions of $v$ where the word labelling the path from the root to the respective node occurs.
Each of these bit-sets, of size $\bigo(\log n)$ bits, can be stored in constant space;
indeed, each $\log n$ block of bits can be seen in fact as a number between $1$ and $n$ so we only need to store a constant number of numbers smaller than $n$;
in our model, each such number fits in a memory word.
Computing the bit-sets can be done in a bottom up manner in linear time: for a node, we need to make the union of the bit-sets of its children, and this can be done by doing a bitwise {\bf or} operation between all the bit-sets corresponding to the children.
So, now, checking the bit-set associated to the lowest node of the suffix tree such that the label of the path from the root to that node starts with $v[j\cdot 2^k+1,(j+1)2^{k}]$ we can immediately output a representation of this factor's occurrences in $v$.

This concludes the proof for $\alpha=1$. 

For $\alpha>1$, we just have to repeat the algorithm in the previous proof for the words $v[i\log n+1,(i+2)\log n]v[(\alpha-1)\log n+1 ,\alpha\log n]$, for $0\leq i\leq \alpha-2$, which allows us to find all the occurrences of the basic factors of $v[(\alpha-1)\log n+1 ,\alpha\log n]$ in $v$. The time is clearly $\bigo(\alpha)$. 
\end{proof}

\begin{remark}\label{find_occ_small_range}
By this lemma, given a word $v$, $|v|=\alpha \log n$,  and a basic factor $y=v[j\cdot2^k+1,(j+1)2^{k}]$, with $j,k\geq 0$ and $j2^k+1> (\alpha-1)\log n$, we can produce $\bigo(\alpha)$ bit-sets, each containing exactly $\bigo(\log n)$ bits, characterising all the occurrences of $y$ in $v$. Let us also assume that we have access to all values $\log x$ with $x\leq n$ (which can be ensured by a $\bigo(n)$ preprocessing). Now, using the bit-sets encoding the occurrences of $y$ in $v$ and given a factor $z$ of $v$, $|z|=c|y|$ for some $c\geq 1$, we can obtain in $\bigo(c)$ time the occurrences of $y$ in $z$: the positions (at most $c$) where $y$ occurs outside a run and/or at most $c$ runs containing the occurrences of $y$. Indeed, the main idea is to select by bitwise operations on the bit-sets encoding the factors of $v$ that overlap $z$ the positions where $y$ occurs (so the positions with an $1$). For each two consecutive such occurrences of $y$ we detect whether they are part of a run in $v$ (by $\LCP$-queries on $v$) and then skip over all the occurrences of $y$ from that run (and the corresponding parts of the bit-sets) before looking again for the $1$-bits in the bit-sets.
 \end{remark}
 \begin{proof}
When given an additional subrange $z$ of $v$, we have to select from the non-null bits of the bit-sets corresponding to the factors $v[j\log n +1 ,(j+2)\log n]$ of $v$ those corresponding to the range $z$. This can be easily done with bitwise operations in $\bigo(c)$ time. Indeed, when looking at such a bit-set, we first select the most significant $1$-bit (it occurs on the position given by $\log m$, where $m$ is the integer value of the current bit-set), it gives an occurrence of $y$. Then we turn this bit into $0$, and select the next most-significant $1$, which again gives an occurrence of $y$. If these two occurrences are part of a run, we detect the ending position of the run by a $\LCP$ query on $v$, and continue finding the $y$'s after that position. Otherwise, we repeat the procedure (the second $y$ detected able will be now detected as the first $y$).
\end{proof}

\bigskip

\noindent{\bf Missing cases from the proof of Lemma \ref{short_reps}}.

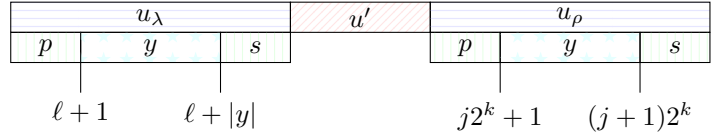
\begin{figure}[h]
	\centering{%
		\begin{tikzpicture}[yscale=0.4,xscale=0.92]
  \col[arm]{4}{$\l\arm$}
  \col[gap]{2}{$u'$}
  \col[arm]{4}{$\r\arm$}
  \newrow
  \col[subarm]{1}{$p$}
  \col[diff]{2}{$y$}
  \col[subarm]{1}{$s$}
  \void{2}
  \col[subarm]{1}{$p$}
  \col[diff]{2}{$y$}
  \col[subarm]{1}{$s$}
  \newrow
  \void{1}
  \lab{$\ell+1$}
  \void{2}
  \lab{$\ell+\abs{y}$}
  \void{4}
  \lab{$j2^k+1$}
  \void{2}
  \lab{$(j+1)2^k$}
	\end{tikzpicture}
	}
	\caption{Fixing $x_m$ in the proof of \Cref{short_reps}, we try to spot gapped repeats whose arms contain a certain basic factor.
		If we can extend this gapped repeat to a maximal gapped repeat, we output it.
	}
	\label{figExtendXm}
\end{figure}

\begin{figure}[h]
	\centering{%
\subfloat[ ]{%
		\begin{tikzpicture}[yscale=0.4,xscale=0.5]
  \col[rep]{1}{$\l\rep$}
  \col[diff]{0.5}{$\subarm$}
  \void{0.5}
  \col[rep]{2}{$\r\rep$}
  \col[diff]{0.5}{$\subarm$}
  \newrow
  \col[arm]{1.5}{$\l\arm$}
  \void{1.5}
  \col[arm]{1.5}{$\r\arm$}
	\end{tikzpicture}
}
\hfill
\subfloat[ ]{%
		\begin{tikzpicture}[yscale=0.4,xscale=0.5]
  \col[rep]{2}{$\l\rep$}
  \void{0.5}
  \col[rep]{1.5}{$\r\rep$}
  \newrow
  \col[arm]{1}{$\l\arm$}
  \void{2}
  \col[arm]{1}{$\r\arm$}
  \newrow
  \void{0.5}
  \col[arm]{1}{$\l\arm$}
  \newrow
  \void{0.5}
  \lab{$p$}
  \void{0.5}
  \col[arm]{1}{$\l\arm$}
	\end{tikzpicture}
}
\hfill
\subfloat[ ]{%
		\begin{tikzpicture}[yscale=0.4,xscale=0.5]
  \col[rep]{1}{$\l\rep$}
  \void{1}
  \col[rep]{2}{$\r\rep$}
  \newrow
  \col[arm]{1}{$\l\arm$}
  \void{1.5}
  \col[arm]{1}{$\r\arm$}
  \newrow
  \col[subarm]{0.5}{$\zarm$}
  \col[subarm]{0.5}{$\zarm$}
  \void{1}
  \col[subarm]{0.5}{$\zarm$}
  \newrow
  \void{0.5}
  \lab{$p$}
	\end{tikzpicture}
}
\hfill
\subfloat[ ]{%
		\begin{tikzpicture}[yscale=0.4,xscale=0.5]
  \col[rep]{2}{$\l\rep$}
  \void{1}
  \col[rep]{2}{$\r\rep$}
  \newrow
  \void{0.5}
  \col[arm]{1.5}{$\l\arm$}
  \void{1}
  \col[arm]{1.5}{$\r\arm$}
	\end{tikzpicture}
}

\subfloat[ ]{%
		\begin{tikzpicture}[yscale=0.4,xscale=0.5]
  \col[rep]{2}{$\l\rep$}
  \void{0.5}
  \col[rep]{1}{$\r\rep$}
  \newrow
  \void{0.5}
  \col[arm]{1}{$\l\arm$}
  \void{1}
  \col[arm]{1}{$\r\arm$}
	\end{tikzpicture}
}
\hfill
\subfloat[ ]{%
		\begin{tikzpicture}[yscale=0.4,xscale=0.5]
  \col[rep]{2}{$\l\rep$}
  \col[diff]{0.5}{$\subarm$}
  \void{0.5}
  \col[rep]{1.5}{$\r\rep$}
  \col[diff]{0.5}{$\subarm$}
  \newrow
  \void{0.5}
  \col[arm]{2}{$\l\arm$}
  \void{0.5}
  \col[arm]{2}{$\r\arm$}
	\end{tikzpicture}
}
\hfill
\subfloat[ ]{%
		\begin{tikzpicture}[yscale=0.4,xscale=0.5]
  \col[diff]{0.5}{$\subarm$}
  \col[rep]{2}{$\l\rep$}
  \void{1}
  \col[diff]{0.5}{$\subarm$}
  \col[rep]{2}{$\r\rep$}
  \newrow
  \col[arm]{2}{$\l\arm$}
  \void{1.5}
  \col[arm]{2}{$\r\arm$}
	\end{tikzpicture}
}
}
\caption{Catching gapped repeats with periodicity is done by case analysis in the proof of \Cref{short_reps}. 
	Each case is depicted in order (from left to right, top to bottom)}
\label{figPeriodicExtension}
\end{figure}
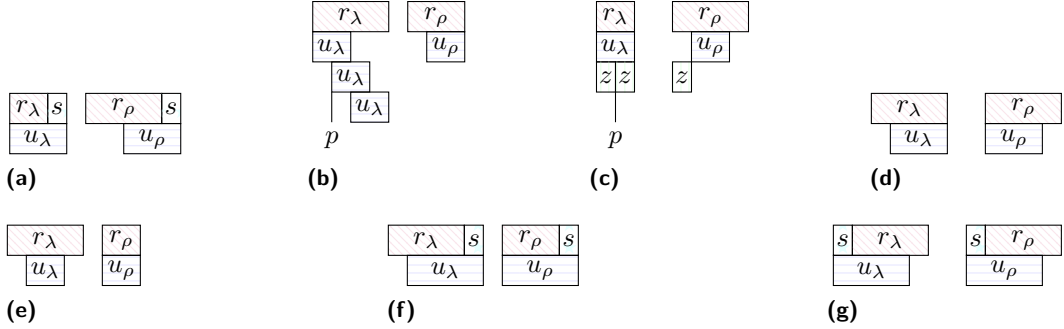

Assume $u_\rho$ starts on the first position of $r_\rho$. If $u_\rho$ ends on a position inside $r_\rho$, other than its last position, then $u_\lambda$ should end on the last position of $r_\lambda$ (otherwise, both arms could be extended to the right). This means that we know exactly the gap between the two arms of the $\alpha$-gapped repeat that we want to construct, as well as the fact that the arms are $p$-periodic. We consider, again, the longest $p$-periodic suffix of $r_\lambda$ which is a prefix of $r_\rho$, and see whether the two occurrences of this factor determine a maximal $\alpha$-gapped repeat. If $u_\rho$ ends on the last position of $r_\rho$ then we know the exact arm of the $\alpha$-gapped palindrome we are looking for (that is, if it also fulfils the length conditions). We can proceed just like in the case analyzed before, when we knew that $u_\lambda=r_\lambda$, only that in this case we have to determine the positions where $u_\lambda$ may start instead of those where $u_\rho$ started. Finally, if $u_\rho$ ends on a position to the right of $r_\rho$, then $u_\lambda$ should also end after $r_\lambda$ and the suffix of $u_\lambda$ occurring after $r_\lambda$ should be equal to the suffix of $u_\rho$ occurring after $r_\rho$. We determine this suffix by a longest common prefix query on $x_m$, and then we obtain exactly the arms of the repeat. We can check whether it is a valid maximal $\alpha$-gapped repeat, and, if so, return it.

The last case is when $u_\rho$ starts on a position to the left of $r_\rho$. Then $u_\lambda$ starts on a position occurring before the first position of $r_\lambda$, the prefix of $u_\rho$ occurring before the beginning of $r_\rho$ equals the prefix of $u_\lambda $ occurring before $r_\lambda$, and they can be determined in constant time. Thus we know the starting points of both arms of the repeat, and we can determine them exactly by a longest common prefix query. We just have to check whether the arms we obtained form a valid maximal $\alpha$-gapped repeat.

\end{document}